\documentclass[10pt,a4paper]{article}
\usepackage[english]{babel}
\usepackage{amssymb,amsmath, amsthm}
\usepackage{graphicx}
\usepackage[table]{xcolor}

\newcommand{\ins}{\underset{\text{\scriptsize S}}{\overset{\text{\scriptsize ins}}{\Longrightarrow}}}
\newcommand{\del}{\underset{\text{\scriptsize S}}{\overset{\text{\scriptsize del}}{\Longrightarrow}}}
\newcommand{\insdel}{\underset{\text{\scriptsize S}}{\overset{}{\Longrightarrow}}}

\newcommand{\insa}[2][\,]{\underset{\text{\scriptsize #2}}{\overset{\text{\scriptsize ins$_{#1}$}}{\Longrightarrow}}}
\newcommand{\dela}[2][\,]{\underset{\text{\scriptsize #2}}{\overset{\text{\scriptsize del$_{#1}$}}{\Longrightarrow}}}

\newcommand\y{\cellcolor[gray]{0.8}}

\newcommand{\N}{\mathbb N}

\newcommand{\RW}{\mathbb{RW}}

\newcommand{\FDRW}{\mathbb{FDRW}}

\theoremstyle{plain}
\newtheorem{theorem}{Theorem}
\newtheorem{example}{Example}
\newtheorem{lemma}[theorem]{Lemma}
\newtheorem{corollary}[theorem]{Corollary}

\newtheorem{claim}[theorem]{Claim}
\newtheorem{definition}{Definition}

\begin{document}

\title{On insertion-deletion systems over relational words}

\author{Igor Potapov\\
\small {\it potapov@liverpool.ac.uk}\\
\small Department of Computer Science, University of Liverpool\\
\small Ashton Building, Ashton Street, Liverpool L69 3BX, UK\\\\
Olena Prianychnykova\\
\small {\it olena.prian@tu-ilmenau.de}\\
\small Fachgebiet Automaten und Logik, Instituts f\"ur Theoretische Informatik,\\
\small Technische Universit\"at Ilmenau,
\small D-98684, Ilmenau, Germany\\\\
Sergey Verlan\\
\small {\it verlan@u-pec.fr}\\
\small LACL, Departement Informatique,  Universit\'e Paris Est Cr\'eteil\\
\small  61, av. General de Gaulle, 94010 Cr\'eteil, France}

\maketitle

\begin{abstract}

We introduce a new notion of a relational word as a finite totally ordered
set of positions endowed with three binary relations that describe which
positions are labeled by equal data, by unequal data and those having an
undefined relation between their labels. We define the operations of
insertion and deletion on relational words generalizing corresponding
operations on strings.  We prove that the transitive and reflexive closure of
these operations has a decidable membership problem for the case of short
insertion-deletion rules (of size two/three and three/two). At the same time,
we show that in the general case such systems can produce a coding of
any recursively enumerable language leading to undecidabilty of reachability questions.

{\bf AMS Subject Classification:} F.4.2 Grammars and Other Rewriting Systems, F.4.3 Formal Languages

{\bf Keywords:} Infinite alphabet, relational words, insertion-deletion, membership 

 \end{abstract}

\section{Introduction}
Nowadays there is a sufficiently broad research activity in the area of
logic and automata for words and trees over infinite alphabets. It is mainly motivated by the
need to analyse and verify infinite-state systems, which for example can use infinite alphabet
of natural numbers {1, 2, 3, . . .} instead of finite number of symbols like {a, b, c}. In the
seminal paper of M. Kaminski and N. Francez \cite{Kaminski}  a very restricted memory structure of the
automaton (Register Automaton) working with words over infinite alphabets was introduced.
The register automaton is operating by keeping a finite number of symbols (from the working
tape) in its memory and making their comparison to other observed symbols. The model
allows recognising a large class of languages over infinite alphabet and at the same time is
not taking any advantage of its memory capabilities beyond what is needed for that purposes.
Later, more models  including automata on data words,  data trees,  pebble automata, etc
have been considered in the realm of semistructured data, timed automata and extended temporal logics \cite{Neven, Segoufin}.
Following motivation from program verification, analysis of XML query languages and other systems
operating explicitly with data values, the research was focused on characterization of automata models and
logics manipulating data in terms of expressive power and decidability.

Comparing to language recognizers,  more complex systems operating with words over infinite alphabet
may require updating them in addition to the operations of comparison between symbols. Obviously, unrestricted
and very general rules allowing rewriting over arbitrary infinite alphabet are too powerful
making most of the computational problems to be undecidable \cite{Bouajjani_FCT2007}. On the other hand there are
existing fragments of rewriting systems over infinite alphabet with a decidable word problem
(i. e. the algorithmic problem of deciding whether two given representatives represent the same
element of the set). These examples are not limited to classical computer science objects, but
also include examples from other areas. One of such examples is unknotedness and equivalence of knots, where
words over infinite alphabet are Gauss words (or Gauss diagrams) and the system of rewriting
rules is a set of Reidemeister moves represented by insertion/deletion and swapping some of
the symbols on Gauss words \cite{Saleh}. While the set of the Reidemeister moves is quite powerful the
word problem for such rewriting rules on Gauss words is decidable following algorithms from
combinatorial topology.

In this paper we aim to extend the concept of the computations on words over infinite alphabet
but preserving original idea of indirect references, i.e. computations where we only make
comparison between positions in our data without explicit references to their values.
In particular we extend the notion of a word on an infinite alphabet by allowing the
equivalence relation to be defined on a subset of the set of positions of a word.
A new notion of {\sl relational word} is defined as a finite set of positions equipped with binary relations
that describe which positions are labeled by equal and non equal data, while for some pairs of
positions the relation between their labels could be undefined. Similar idea of representing data over
a finite alphabet as a set of relations was also named as a ``relational code'' can be found in \cite{HHK2007},
which generalize ``partial words'' in the area of nonstandard stringology \cite{Muthukrishnan} and
DNA sequence processing \cite{Leupold}.
Another example can be found in \cite{Bojanczyk} where authors introduce a first-order logic $FO(\sim,<,+1)$ , where
for every formula $\varphi$ in this logic the set $L(\varphi)$ is the set of data words that satisfy a sentence $\varphi$.
This approach also allows the specification using a kind of a template, the main difference from the model proposed
in the present paper being the unbounded length of the specified string (in the case of the model from this paper
all the words corresponding to the same relational word have the same fixed length).
All of the above models study words and languages described by some relations, but do not define any rewriting on
these structures.

In the concept of {\sl relational word}, instead of a particular words over an infinite alphabet,
we can operate with templates that may represent finite or
infinite languages depending on a choice of the alphabet.  This gives us an opportunity
to define rewriting of data on a new conceptual level focusing only on the operations of
rewriting based on existing relations in data and abstracting from actual data on which we may operate.

The rewriting system on data is interesting both from theoretical and practical aspects, see \cite{MEMICS2009,Bouajjani_FCT2007,ICDT2014}.
In this paper we consider a very natural rewriting system, motivated by the knot theory, in which only insertion and deletion operations are defined.  Also insertion and deletion are considered to be the basic operations in DNA processing and RNA editing \cite{Petre} and in the context of an infinite alphabet insertion-deletion systems are important for reasoning
about recursive sequential programs, multithreaded programs, parametrized and dynamic networks of processes, etc \cite{Bouajjani_FCT2007}.
In particular we study the membership problem: for a given set of insertion/deletion operations defined by relational words
decide whether a relational word $w$ can be derived from an empty word. We show that for any system which
inserts 2 symbols and deletes 3 symbols or vice versa the membership is decidable. We also
show that this does not hold anymore for longer insertion and deletion rules -- in this case the membership and the word problems are undecidable.

\section{Preliminaries}
\subsection {Relational words}
A finite sequence of elements of a finite alphabet $\Sigma$ is called a finite word over $\Sigma$, or just a word. We denote by $\Sigma^*$
the set of words over $\Sigma$ and by $\Sigma^+$ the set of
nonempty words. The empty word is denoted by $\varepsilon$.

Let $\Delta$ be an infinite set. A word over an infinite alphabet $\Delta$ is a
finite sequence of elements of $\Delta$ \cite{Kaminski,Manuel,Neven,Segoufin}.
Elements of a finite alphabet $\Sigma$ are defined explicitly and could be
accessed directly, while elements of an infinite alphabet $\Delta$ could be
only tested for equality. Then a word over an infinite alphabet may be viewed
as a finite totally ordered set of positions endowed with an equivalence
relation.

A well-known example of words over an infinite alphabet are data
words\cite{Manuel,Segoufin}. Let $\Sigma$ be a finite alphabet and $D$ be an
infinite set of data values. A data word is a finite sequence over $\Sigma
\times D$, i.e., in a data word each position carries a label from a finite
alphabet and a data value from some infinite domain. A data word may be viewed
as a word over finite alphabet $\Sigma$ with an equivalence relation on the set
of its positions \cite{Bojanczyk}.

Now the idea of this paper is to extend the notion of a word over an infinite
alphabet by allowing the equivalence relation to be defined on a subset of the
set of positions of the word. We define a relational word as a finite set of
positions equipped with  binary relations that describe which positions are
labeled by equal and by inequal data, while for some pairs of positions the
relation between their labels is not defined.

A relational word can be viewed as a kind of a template. For an alphabet $A$
(finite or infinite) a relational word $W$ defines a language $L_A(W )\in A^*$
which is the set of all words $w=a_1 a_2 .... a_n$,  where $a_i\in A$, $1\le
i\le n$, with $n$ being the length of $W$, such that for every pair of
positions $i$ and $j$ in $W$ we have
\begin{itemize}
  \item if $(i,j)$ belongs to the equality relation, then $a_i=a_j$
  \item if $(i,j)$ belongs to the inequality relation, then $a_i \neq a_j$
\end{itemize}

We remark that if any pair of positions of a relational word $W$ is a member of
a relation (equality or inequality), then $L_A(W)$ can be identified with any
element $w\in L_A(W)$, as based on $w$ and $A$ it is possible to reconstruct
$L_A(W)$. This gives the possibility to represent a relational word using
ordinary words, where same letters represent the equality relation and
different letters the inequality one. If it will be clear from the sequel we
will not indicate the index $A$ in $L_A(W)$.

With every relational word $W$ we can associate a graph $G^W=(Q,T)$ and an edge
labeling function $Lab_{G^W}: T\to \{0,1\}$ such that
\begin{itemize}
  \item $Q=\{q_1,q_2,...,q_n\}$ is an ordered set of nodes, $n$ is the length
      of $W$,
  \item $T\subseteq Q\times Q$ is the set of edges such that $(q_i,q_j)\in T$
      iff there is a relation (equality or inequality) between positions $i$
      and $j$.
  \item $Lab_{G^W}$ is defined as follows
      \begin{itemize}
        \item $Lab_{G^W}(q_i,q_j)=1$ if the labels of the positions $i$ and
            $j$ are equal,
        \item $Lab_{G^W}(q_i,q_j)=0$ iff the labels of the positions $i$
            and $j$ are not equal.
      \end{itemize}
\end{itemize}

We will use the following convention for the graphical representation of $G^W$.
The nodes of the graph will be aligned horizontally and the order of nodes
taken from left to right will correspond to their ordering within the graph. We
will depict edges labeled by $1$ below the axis induced by the node alignment
and the edges labeled by $0$ on the top of it. We also note that for any $q_i$,
there exist an edge $(q_i,q_i)$ labeled by $1$. In order to simplify the
pictures we will not draw corresponding self-loops.

With every relational word $W$ we associate the matrix $M^W\in {\{0,1,2\}}^{n \times n}$ where  $n$ is the length
      of $W$, as follows:\\
{\small      $M^W[i,j]= \begin{cases} 1 & \text{ iff  the labels of the
positions $i$ and
            $j$ are equal} \\
0 & \text{ iff the labels of the positions $i$
            and $j$ are not equal} \\
2 & \text{ iff the relation between the labels of the positions $i$
            and $j$ is not defined}
\end{cases}$
}
\begin{example}
Let us consider the relational word  $W$ of length $4$ such that the labels of
the first and the third position are equal, the label of the second position is
not equal to them, and the relations of the label of the fourth position to all
others are undefined. The graph that represents $W$ and the corresponding
matrix are shown on the Fig.~\ref{fig:ex1}.
\begin{figure}[h]
\begin{center}
\hfil\parbox[c]{5cm}{\includegraphics[scale=0.5]{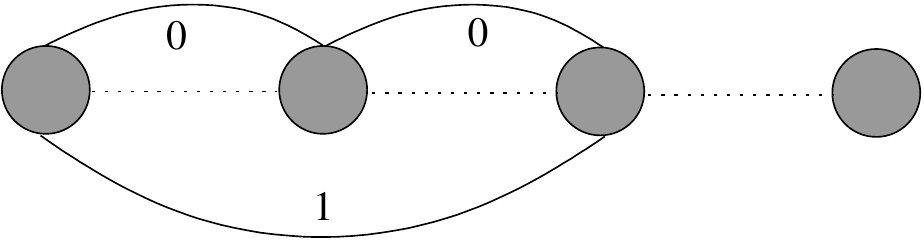}} \hfil
\parbox[c]{5cm}{\footnotesize $M^W=\left(\begin{array}{cccc}
         1 & 0 & 1 & 2 \\
         0 & 1 & 0 & 2 \\
         1 & 0 & 1 & 2 \\
         2 & 2 & 2 & 1 \\
       \end{array}
     \right)$}
     \hfil
\caption{An example of a relational word and the corresponding matrix.}\label{fig:ex1}
\end{center}
\end{figure}

\noindent Let $A=\{a\}$ , then $L(W)=\emptyset$.\\
Let $A=\{a,b\}$, then $L(W)=\{abaa, abab, baba, baba\}$.\\
Let $A=\{a,b,c\}$, then $L(W)=\{abaa; abab; abac; baba; babb; babc; acaa; acab;\\ acac; caca; cacb; cacc; bcba; bcbb; bcbc; cbca; cbcb; cbcc \}$.\\
\end{example}

 Let $W$ be a fully defined relational word and $A$ be an
alphabet. Then every word $w\in L(W )$ is an assignment of equivalence classes
to symbols of $A$. Hence if the alphabet $A$ is finite, then the language $L(W
)$ is finite (and thus it is regular). Moreover,  $|L(W)|=C^k_n$ where $k$ is
the number of classes of the equivalence relation and $n$ is the size of the
alphabet $A$. If the relational word $W$ is not fully defined, then for every
$w\in L(W )$ the number $k$ of different symbols in $w$ could not be larger
than the length of the word $W$. Then for every alphabet $A$ and every
relational word $W$ the set $L(W )$ is finite if and only if $A$ is finite.

Now we give a formal definition of a relational word.

\begin{definition}
A \em{relational word} is a relational structure $W=(X^W, E^W, N^W)$ where
\begin{itemize}
  \item $X^W=(X^W, \prec )$ is a finite totally ordered set;
  \item $E^W$ and $N^W$ (for equal and not equal) are binary relations on
      $X^W$ such that
      \begin{itemize}
        \item they are mutually exclusive: $E^W\cap N^W=\emptyset$;
        \item $E^W$ is an equivalence relation;
        \item $N^W$ is a symmetric relations;
        \item for every $x,y,z \in X^W$,  if $(x,y)\in E^W$, then $(x,z)\in
            R^W \text{if } (y,z)\in R^W, R\in\{E, N\}$.
      \end{itemize}
\end{itemize}
\end{definition}

For technical reasons we shall consider the relation $U^W=X^W\times
X^W\setminus (E^W\cup N^W)$ corresponding to an undefined relation between
pairs of positions.

We denote by $|W|=|X^W|$  the length of the relational word $W$ and by $W[i]$
the $i$-th element from the ordering of $X^W$. The empty relational word is
denoted by $\varepsilon$, $|\varepsilon|=0$. A relational word $W$ is fully
defined if $U^W=\emptyset$.

We denote \emph{the set of all relational words} by $\mathbb{RW}$ and \emph{the set of all fully defined relational words} is denoted by $\mathbb{FDRW}$.

\begin{example}
Let us consider the relational word $W$ from the Example 1. We have that $X^W=\{x_1, x_2, x_3, x_4\}$; $x_1\prec x_2\prec x_3\prec x_4$; $x_1$ is equal to $x_3$, $x_2$ is not equal to $x_1$ and $x_3$, the relations between $x_4$ and $x_1, x_2, x_3$ are undefined, i.e.,
\begin{itemize}
  \item $E^W=\{(x_1, x_1), (x_2, x_2), (x_3, x_3), (x_4, x_4), (x_1, x_3), (x_3, x_1) \}$;
  \item $N^W=\{(x_1, x_2), (x_2, x_1), (x_2, x_3), (x_3, x_2)\}$;
\end{itemize}
\end{example}


\begin{definition}
Two relational words $W$ and $V$ are equal if $|W|=|V|=n$,  and for every
$1\leq i,j \leq n$ we have {\small $(i, j)\in \begin{cases}
E^V & \text{ iff }(i, j)\in E^W, \\
N^V & \text{ iff } (i, j)\in N^W.\
\end{cases}$}

\end{definition}

We also introduce the notion of \emph{contradiction} for relational words.
Informally, word $W$ contradicts  word $V$ if it is impossible to get the same
fully defined relational word by instantiating to equality or inequality the
relations between undefined positions of $W$ and $V$. If $W$ contradicts $V$,
then $L(W)\cap L(V)=\emptyset$. Formally we define contradiction as follows.

\begin{definition}
A relational word $W$ contradicts a relational word $V$ if $|W|=|V|=n$,  and there are $i$ and $j$ such that $1\leq i,j \leq n$ and either  $(i, j)\in E^W$ and $(i, j)\in N^V$, or $(i, j)\in N^W$ and $(i, j)\in E^V$.
\end{definition}

In order to be able to work not only with a relational word as a whole but with parts of it, we will need a notion of a subword.

\begin{definition}
A relational word $W$ is a scattered subword of $V$  if $X^W\subseteq X^V$ and
for every $x,y \in X^W$ we have {\small $(x,y)\in  \begin{cases}
E^W & \text{ iff }(x,y)\in E^V,\\
N^W& \text{ iff } (x,y)\in N^V.\\
\end{cases}$
}

A relational word $W$ is a \em{subword} of $V$ if it is a scattered subword of $V$ and for every $x,y \in X^V$ if there are $y,z\in X^W$ such that $y\prec x \prec z$, then $x\in X^W$.
\end{definition}

\begin{example}\label{ex:sw} Fig.~\ref{fig:sw} depicts the above notions.

\begin{figure}[ht]
\hfill
  \begin{minipage}{0.34\textwidth}
    \includegraphics[scale=0.3]{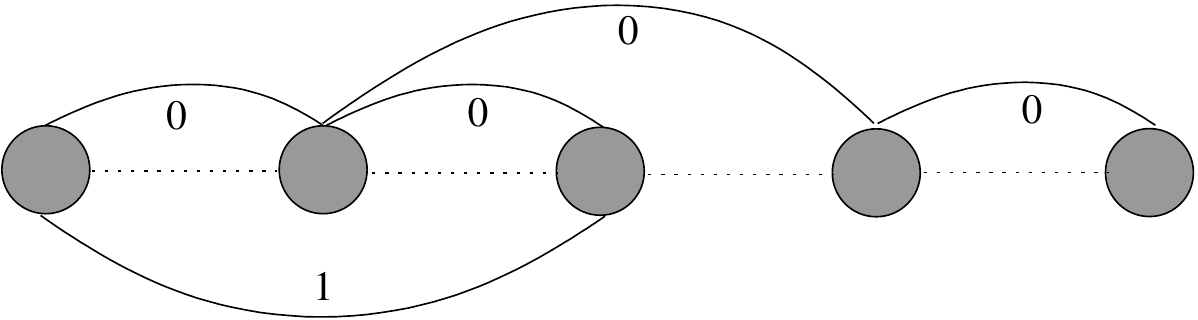}
  \\(a)  relational word $W$ \\
  \end{minipage}
  \hfill
  \begin{minipage}{0.33\textwidth}
     \includegraphics[scale=0.3]{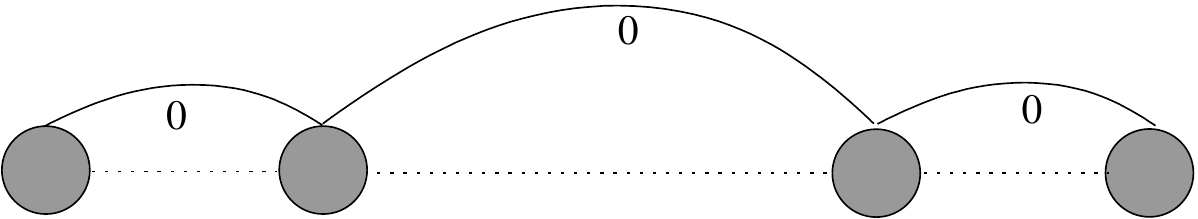}
   \\  (b) scattered subword of $W$\\
  \end{minipage}
  \hfill
  \begin{minipage}{0.3\textwidth}
    \includegraphics[scale=0.3]{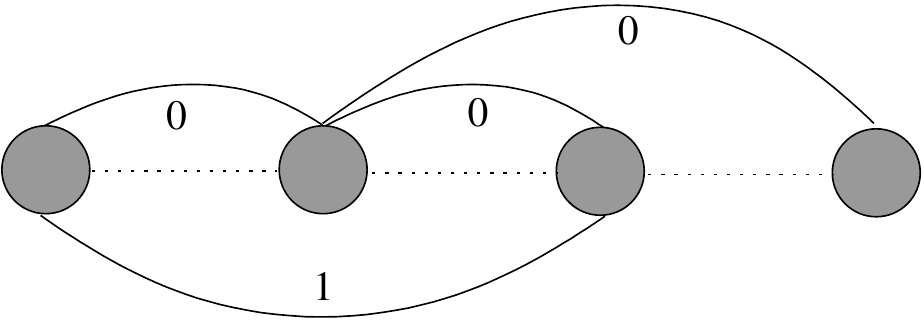}
 \\(c) subword  of $W$\\
  \end{minipage}
\hfill
\vspace{-3mm}
   \caption{An example of a  relational word with its scattered subword and
   subword}\label{fig:sw}
\end{figure}
\end{example}

\vspace{-4mm}
 With every relational word $W$ we associate its numerical
characteristics:
\begin{enumerate}
  \item $maxFD(W)$ is the length of the longest fully defined scattered subword of $W$,
  \item $maxE(W)$ is  the length of the longest scattered subword $W'$ of $W$ such that every two elements of that subword are equal, i.e., for every $x,y \in X^{W'}$ we have $(x,y)\in E^W$,
\end{enumerate}

\begin{example}
Consider $W$ from the Example~\ref{ex:sw} we have $maxFD(W) = 3$, $maxE(W) =
2$, $maxN(W) = 2$.
\end{example}

\subsection {Insertion-deletion systems on relational words}
\begin{definition}
An insertion-deletion scheme $S$  is a pair $S = (INS, DEL)$ where
$INS\subseteq \mathbb{FDRW}$ is the set of insertion rules and $DEL\subseteq
\mathbb{FDRW}$ is the set of deletion rules \cite{Margenstern,Petre}.
\end{definition}

The insertion-deletion scheme $S = (INS, DEL)$ is called \emph{simple} if it
contains only one insertion rule and only one deletion rule, i.e., $INS=\{I\}$,
$DEL=\{D\}$ where $I,D \in \mathbb{FDRW}$.

We denote  by $I_nD_m$ the set of all simple insertion-deletion schemes such
that the length of the insertion rule is $n$ and the length of the deletion
rule is $m$.

Now we define the operations of insertion and deletion operations on relational
words.

Informally, given $W,V \in \mathbb{RW}$ we understand the single-step insertion
relation $W \ins V$ as follows(Fig.~\ref{fig:ssi}): to obtain $V$, we take $W$
and $Y \in INS$ and ``insert'' $Y$ as a subword between any two symbols of $W$.
We assume that for every pair $(x,y)$, where $x$ is a symbol of $W$ and $y$ is
a symbol of $Y$, the relation between $x$ and $y$ is undefined.

\begin{figure}[ht]
\centering
\includegraphics[scale=0.4]{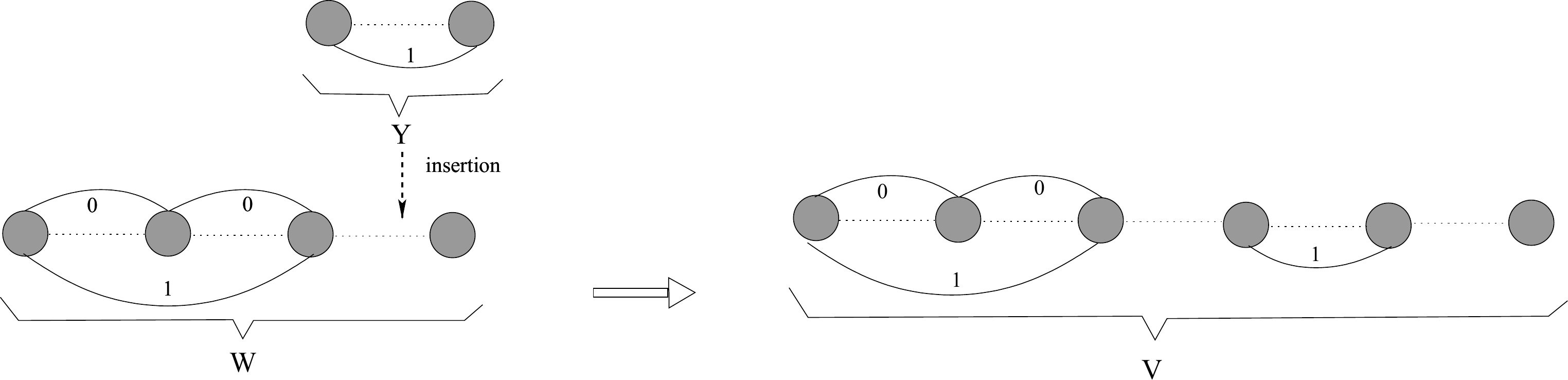}
\caption{The single-step insertion relation $W \ins V$}\label{fig:ssi}
\end{figure}

In matrix notation the operation from Fig.~\ref{fig:ssi} can be represented as: \\
$\left(
       \begin{array}{cccc}
         1 & 0 & 1 & 2 \\
         0 & 1 & 0 & 2 \\
         1 & 0 & 1 & 2 \\
         2 & 2 & 2 & 1 \\
       \end{array}
     \right)   \insa[3]{Y}
     \left(
       \begin{array}{ccc|cc|c}
         1 & 0 & 1 & 2 & 2 & 2\\
         0 & 1 & 0 & 2 & 2 & 2\\
         1 & 0 & 1 & 2 & 2 & 2\\
                    \hline
         2 & 2 & 2 & 1 & 1 & 2\\
         2 & 2 & 2 & 1 & 1 & 2\\
                    \hline
         2 & 2 & 2 & 2 & 2 & 1\\
       \end{array}
     \right)
      $, where $Y= \left(
       \begin{array}{cc}
         1 & 1 \\
         1 & 1 \\
       \end{array}
     \right) $.

Formally, we define this relation as follows.

Consider the function $s_{k,m}:\N\to \N$, $k,m\in\N$ defined as follows
$$\small s_{k,m}(i)=\begin{cases}
i & \text{ if } 1\le i\le k,\\
i+m & otherwise
\end{cases}
$$

\begin{definition}\label{def:insertion}

The single-step insertion relation on  $\mathbb{RW}$ that is induced by $S =
(INS, DEL)$ is defined as follows. For any $W,V \in \mathbb{RW}$, $Y\in INS$
and an integer $0\le k\le |W|$ we have $W \insa[k]{Y} V$  iff
\begin{itemize}
  \item $(i,j)\in R^W$ implies $(s_{k,m}(i),s_{k,m}(j))\in R^V$, where
      $m=|Y|$ and $R\in \{E,N\}$,
  \item $(i,j)\in R^Y$ implies $(i+k,j+k)\in R^V$, where $1< i,j\le |Y|$,
      $R\in \{E,N\}$.
\end{itemize}
\end{definition}

If we are not interested by the site of the insertion or by the concrete
insertion rule then we will write $W\ins V$, meaning that there exists $Y\in
INS$ and $k\ge 0$ such that $W\insa[k]{Y}V$.

\begin{definition}
The insertion relation on $\mathbb{RW}$ that is induced by $S = (INS, DEL)$ is the reflexive, transitive closure of $\ins$ and is denoted by $\ins^*$.
\end{definition}

Now we explain the deletion relation. Informally, the application of the
deletion rule $W \del V$ consists of two steps: expansion and deletion
(Fig.~\ref{fig:ssd}). First, we have to find a subword $Y'$ in the relational
word $W$ that does not contradict to a relational word $Y\in DEL$ and to
``expand'' it to $Y$: for every symbols $x$ and $y$ in $Y'$ such that the
relation between them is undefined, we set this relation to be the same as the
relation between the corresponding symbols of $Y$ (a thick line on
Fig.~\ref{fig:ssd}). In order to preserve transitivity, if we define that $x$
is equal to $y$, then we have to connect to $x$ all nodes incoming to $y$ and
using the same label (dotted lines on Fig.~\ref{fig:ssd}). Next, we take the
``expanded'' subword out of the word $W$ and obtain the word $V$.
\begin{figure}[ht]
\centering
\includegraphics[scale=0.3]{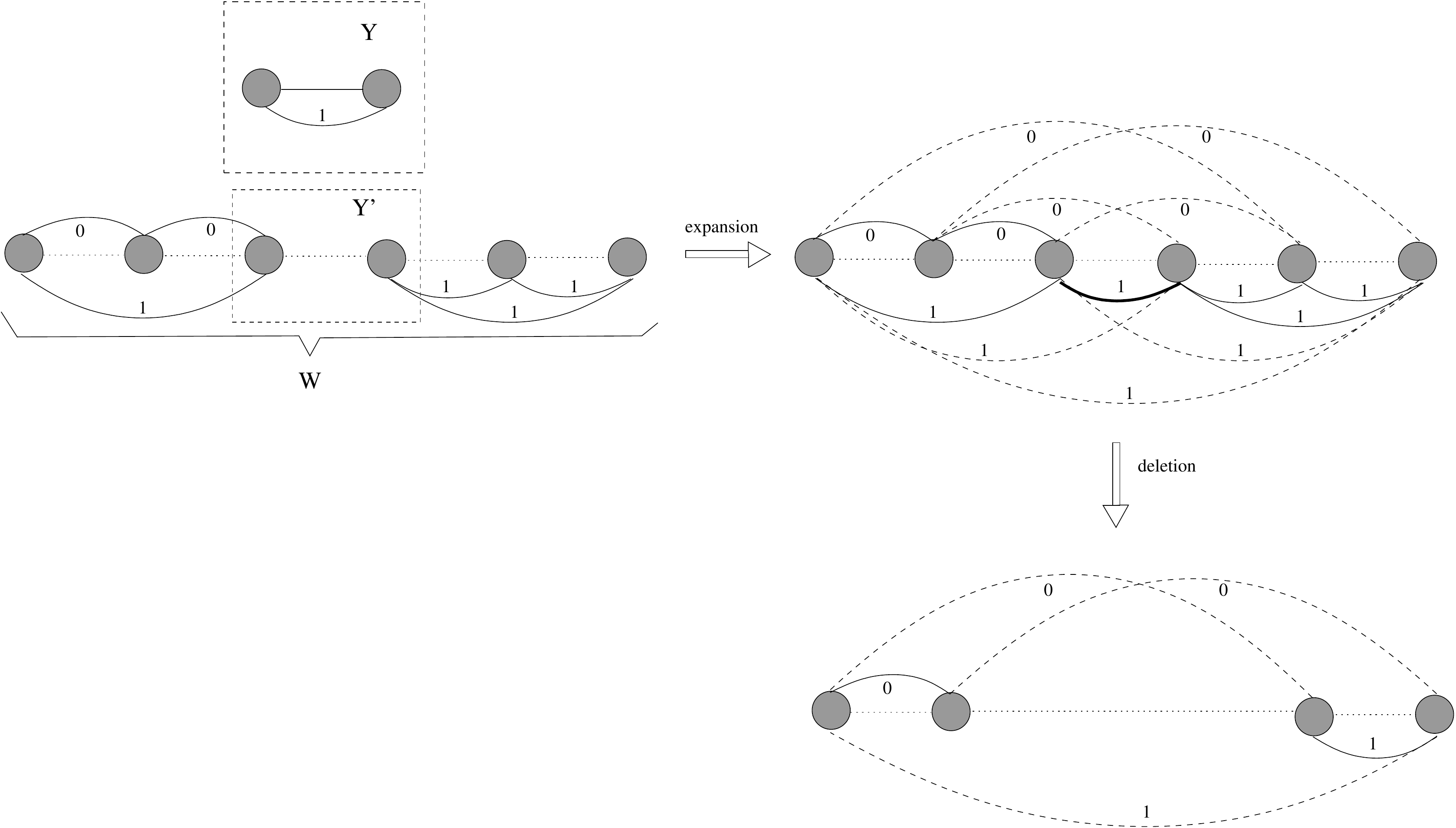}
\caption{The single-step deletion relation $W \del V$}\label{fig:ssd}
\end{figure}
In matrix notation the operation from Fig.~\ref{fig:ssd} can be represented as: \\
$ \left(
       \begin{array}{cc|cc|cc}
         1 & 0 & 1 & 2 & 2 & 2\\
         0 & 1 & 0 & 2 & 2 & 2\\
          \hline
         1 & 0 & \y1 & \y2 & 2 & 2\\
         2 & 2 & \y2 & \y1 & 1 & 1\\
          \hline
         2 & 2 & 2 & 1 & 1 & 1\\
         2 & 2 & 2 & 1& 1 & 1\\
       \end{array}
     \right) \dela[3]{Y}
      \left(
       \begin{array}{cccc}
         1 & 0 &  0 & 1\\
         0 & 1 &  2 & 0\\
         0 & 2 &  1 & 1\\
         1 & 0 &  1 & 1\\
       \end{array}
     \right)
$, where $Y= \left(
       \begin{array}{cc}
         1 & 1 \\
         1 & 1 \\
       \end{array}
     \right) $.

Formally the single-step deletion relation is defined as follows.

\begin{definition}\label{def:deletion}
The single-step deletion relation on $\mathbb{RW}$ that is induced by $S =
(INS, DEL)$ is defined as follows. For any $W,V \in \mathbb{RW}$, $Y \in Del$
and a integer $1\le k\le |W|$ we have $W \dela[k]{Y} V$ if
\begin{itemize}
\small
  \item $(i,j)\in R^W$ implies $(s^{-1}_{k-1,m}(i),s^{-1}_{k-1,m}(j))\in
      R^V$, where $m=|Y|$ and $R\in \{E,N\}$,
  \item $(i,j)\in R^Y$ implies
    \begin{itemize}
      \item either $(i+k-1,j+k-1)\in R^W$, where $1< i,j\le |Y|$, $R\in
          \{E,N\}$,
      \item or $(i+k-1,j+k-1)\in U^W$ and for all pairs $(i+k-1,q)\cup
          (p,j+k-1)\in E^W$ it holds
          $(s^{-1}_{k-1,m}(p),s^{-1}_{k-1,m}(q))\in R^V$ ($R\in \{E,N\}$),
      \item or $(i+k-1,j+k-1)\in U^W$, $R=E$ and for all pairs
          $(i+k-1,q)\in E^W$, $(i+k-1,q')\in N^W$,  $(p,j+k-1)\in E^W$,
          $(p',j+k-1)\in N^W$ it holds
          $(s^{-1}_{k-1,m}(p),s^{-1}_{k-1,m}(q))\in N^V$ and
          $(s^{-1}_{k-1,m}(p'),s^{-1}_{k-1,m}(q'))\in N^V$.

     \end{itemize}
\end{itemize}
\end{definition}

As for insertion we will write $W\del V$, meaning that there exists $Y\in DEL$
and $k\ge 1$ such that $W\dela[k]{Y}V$.

\begin{definition}
The deletion relation on $\mathbb{RW}$ that is induced by $S = (INS, DEL)$ is the reflexive, transitive closure of $\del$ and is denoted by $\del^*$.
\end{definition}

Each of the relations $\ins$ and $\del$ is denoted by $\insdel$ and  the reflexive, transitive closure of  $\insdel$ is denoted by $\insdel^*$.

\begin{definition}
An insertion-deletion system is the tuple $S=(V,INS,DEL,A)$, $V$ is an
alphabet, $(INS,DEL)$ is an insertion-deletion scheme and $A\subseteq V^*$ is
the initial language (the axioms) of the system.
\end{definition}

If $A=\emptyset$ then we will use a shorthand notation denoting the
corresponding system as $S=(INS,DEL)$, i.e. we will identify it by the
corresponding insertion-deletion scheme.

\begin{definition}
For an insertion-deletion system $S = (V,INS, DEL,A)$  we define the language
set $L(S)= \{W\in \mathbb{RW}\mid Z \insdel^* W, Z\in A\}$ and the set $FDL(S)=
\{W\in \mathbb{FDRW}\mid Z \insdel^* W, Z\in A\}$.
\end{definition}

\section {Main results}
\begin{lemma} \label{lem:1} 
For every insertion-deletion system $S$ and every $W,V \in \mathbb{RW}$ if $W \insdel^* V$, then there is $Y\in \mathbb{RW}$ such that $W \ins^* Y \del^* V$.
\end{lemma}
\begin{proof}
First we prove that if $W \insdel^* V$ and there are $W_1, W_2, W_3 \in\mathbb{RW}$ such that $W \insdel^* W_1 \del W_2 \ins W_3 \insdel^* V$ then there is ${W_2}'\in \mathbb{RW}$ such that $W \insdel^* W_1 \ins {W_2}' \del W_3 \insdel^* V$, i.e., for any two consecutive operations of deletion and insertion in a derivation we can swap them so the insertion would be performed before the deletion.

Since  $W_1 \del W_2$, by definition of the operation $\del$ we have that there
is a word $Y \in DEL$ and a subword $Y'$ in the word $W_1$ such that $Y'$ does
not contradict $Y$. Let the word $W_1$ has the length $n$, the subword $Y'$ has
the length $m$ and starts from the position $i$ in the word $W_1$. Then
$|W_2|=n-m$. Since $W_2 \ins W_3$, we have that there is a subword $T$ in $W_3$
and there is $T' \in INS$ such that $T=T'$. Let $T$ starts from the position
$j$ in $W_3$. We construct the word $W_2$ as follows. We take the word $W_1$
and insert in it $T'$ starting from the position $j$ if $j<i$ or  from the
position $j + n$ if $j \geq i$. Then $W_1 \ins {W_2}'$.  Then the word ${W_2}'$
has the scattered subword $Y'$ that does not contradict $Y$ since it was not
changed during the insertion. Then we apply the deletion operation to ${W_2}'$
deleting the subword $Y$ . By the definition of the operation $\ins$ all the
relations between symbols in $W_1$ remain the same after insertion,  and all
symbols that have been inserted have undefined relationships with all other
symbols in the word. By the definition of the operation $\del$, during the
first step of deletion when we "expand" the subword $Y$ to match $Y'$, we can
change the relations between only those symbols which has relations with
symbols from $Y$ that are not undefined. Then all relations between symbols
from subword $T$ and all other symbols of the word after expansion remain
unchanged, i.e. undefined, and all other relations are the same as the
relations between corresponding symbols in $W_3$. Then we have that the result
of deleting $Y$ from ${W_2}'$ is equal to $W_3$.

Thus for every derivation we can change $W_1 \del W_2 \ins W_3$ to $W_1 \ins {W_2}' \del W_3$. By repeating this process for any place in the derivation where $\del$ goes directly before $\ins$,  we obtain the new derivation such that all $\ins$ precede all $\del $, i.e., $W \ins^* Y \del^* V$.
\end{proof}

Now we consider only simple insertion-deletion system from  $I_2D_3 \cup I_3D_2$ , i.e., $S = (INS, DEL)$ such that both sets $INS$ and $DEL$ contain only one rule and either the length of the insertion rule is 2 and the length of deletion rule is 3, or the length of the insertion rule is 3 and the length of deletion rule is 2.

Because of the transitivity of the relation $E$, there are only 2 different
fully defined relational words of length 2 and 5  different  fully defined
relational words of length 3, yielding 10 insertion-deletion systems in both
$I_3D_2$ and $I_2D_3$.
Below are the associated matrices.\\
{\scriptsize
 $M_1^2=\left(
   \begin{array}{cc}
     1 & 0 \\
     0 & 1\\
   \end{array}
 \right)
$,
$M_2^2=\left(
   \begin{array}{cc}
     1 & 1 \\
     1 & 1\\
   \end{array}
 \right)
$,
$M_1^3=\left(
  \begin{array}{ccc}
    1 & 1 & 1 \\
    1 & 1 & 1 \\
    1 & 1 & 1 \\
  \end{array}
\right)$,
$M_2^3=\left(
  \begin{array}{ccc}
    1 & 1 & 0 \\
    1 & 1 & 0 \\
    0 & 0 & 1 \\
  \end{array}
\right)$,\\
$M_3^3=\left(
  \begin{array}{ccc}
    1 & 0 & 0 \\
    0 & 1 & 1 \\
    0 & 1 & 1 \\
  \end{array}
\right)$,
$M_4^3=\left(
  \begin{array}{ccc}
    1 & 0 & 1 \\
    0 & 1 & 0 \\
    1 & 0 & 1 \\
  \end{array}
\right)$,
$M_5^3=\left(
  \begin{array}{ccc}
    1 & 0 & 0 \\
    0 & 1 & 0 \\
    0 & 0 & 1 \\
  \end{array}
\right)$. }

\begin{lemma}\label{lem:2} 
For every simple insertion-deletion system $S\in I_2D_3 \cup I_3D_2$ and every relational word $W$ we have $W\insdel^* \varepsilon$.
\end{lemma}
\begin{proof}

First we show that in each $S\in I_3D_2$ we have $V  \insdel^* \varepsilon$ where $|V|=1$.

Let $S=(\{I\}, \{D\})$ where
      {\small
\begin{enumerate}
\setlength\itemsep{-1mm}
  \item $I=\{M_1^3\}$, $D= \{M_1^2\})$. If we take the word $V$, insert the
      word $I$ after the only symbol in $V$,  and  delete two last symbols in
      the result (they are equal), then we obtain the word of length two
      where the relation between symbols is undefined. Then we can delete
      these two symbols as equal and thus get the empty word. Hence we have
      following derivation: $V \insa[1]{I} V_1   \dela[3]{D}  V_2 \dela[1]{D}
      \varepsilon$. The derivations in matrix notation for this case and all
      the following cases can be found in the appendix.
\item   $I=M_1^3$, $D=M_2^2$.
In this case we have $V \insa[0]{I} V_1  \insa[1]{I} V_2 \insa[1]{I} V_3 \dela[7]{D}  V_4 \dela[4]{D} V_5 \dela[3]{D}  V_6 \dela[1]{D} \varepsilon$.\\
\item    $I=M_2^3$, $D=M_1^2$.
In this case we have $V \insa[0]{I} V_1  \insa[2]{I} V_2 \insa[2]{I} V_3 \dela[2]{D}  V_4 \dela[5]{D} V_5 \dela[3]{D}  V_6 \dela[1]{D} \varepsilon$.
\item   $I=M_2^3$, $D=M_2^2$.
 Then we have $V \insa[0]{I} V_1   \dela[1]{D}  V_2 \dela[1]{D} \varepsilon$.
\item   $I=M_3^3$, $D=M_1^2$.
 Then we have $V \insa[0]{I} V_1   \dela[1]{D}  V_2 \dela[1]{D} \varepsilon$.
\item   $I=M_3^3$, $D=M_2^2$.
Then we have $V \insa[0]{I} V_1   \dela[2]{D}  V_2 \dela[1]{D} \varepsilon$.
\item   $I=M_4^3$, $D=M_1^2$.
Then we have $V \insa[0]{I} V_1   \dela[2]{D}  V_2 \dela[1]{D} \varepsilon$.
\item   $I=M_4^3$, $D=M_2^2$.
Then we have $V \insa[0]{I} V_1   \dela[3]{D}  V_2 \dela[1]{D} \varepsilon$.
\item   $I=M_5^3$, $D=M_1^2$.
Then we have $V \insa[0]{I} V_1  \insa[0]{I} V_2 \insa[2]{I} V_3 \dela[6]{D}  V_4 \dela[5]{D} V_5 \dela[2]{D}  V_6 \dela[1]{D} \varepsilon$.
\item   $I=M_5^3$, $D=M_2^2$. Then we have $V \insa[0]{I} V_1   \dela[1]{D}
    V_2 \dela[1]{D} \varepsilon$.
\end{enumerate}
}

Thus in each $S\in I_3D_2$ we have $V  \insdel^* \varepsilon$ .

Since we can delete any isolated symbol, we can apply this sequence of rules to each symbol of the relational word $W$ and thus we can delete the whole word, i.e.,  for each $S\in I_3D_2$ we have that $W\insdel^* \varepsilon$.

It is easy to see that in each $S\in I_2D_3$ where $S=(\{I\}, \{D\})$ the same sequence of rules that was used in $S=(\{D\}, \{I\})$ to delete the isolated symbol but taken backwards adds to a relational word an isolated symbol. Then we can use the following strategy: we add to the end of the relational word $W$ two isolated symbols and then apply to the last three symbols of the result the deletion rule,  thus we obtain $W$ shortened by one symbol. We repeat this process and delete all symbols in $W$ one by one.

Thus we we have that for every   $S\in I_2D_3 \cup I_3D_2$ and every relational word $W$ we have $W\insdel^* \varepsilon$.
\end{proof}
\begin{corollary}\label{cor:3} 
Let $S$ be a simple insertion-deletion system such that $S\in I_2D_3 \cup I_3D_2$. For every fully defined relational words $V$ and $W$ we have  $V \insdel^* W$ iff there is $W'\in \mathbb{RW}$ such that $W$ is a fully defined scattered subword of $W'$ and $V \insdel^* W'$.
\end{corollary}
\begin{proof}
Let $V \insdel^* W$. Since every relation word $W$ has a fully defined subword $W'$ of length 1 and by Lemma~\ref{lem:2} for every $W$ we can delete all the symbols of $W$ but one, then $W \insdel^* W'$.

Let $W\in \mathbb{RW}$  be a fully defined scattered subword of $W'$ and $V \insdel^* W'$. Since by Lemma~\ref{lem:2} we always can delete any symbol from the relation word without changing the relations between all other symbols, we have $V \insdel^* W' \insdel^* W$.
\end{proof}

In the next lemma we analyze the behavior of insertion-deletion systems in
which all symbols in both insertion and deletion rules are equal. We prove that
a relational word could be derived from the empty word if and only if all its
symbols are equal.

\begin{lemma} \label{lem:4} 
Let $S=(\{I\}, \{D\})$ be a simple insertion-deletion system such that $S\in I_2D_3 \cup I_3D_2$ and for every $x,y \in X^I$ we have $(x,y)\in E^I$, for every $x',y' \in X^D$ we have $(x',y')\in E^D$, i.e., all symbols in both insertion and deletion rules are equal. For every relational word $V$ we have $V\in FDL(S)$ iff for every $x,y \in X^V$ we have $(x,y)\in E^V$.
\end{lemma}
To prove Lemma~\ref{lem:4} we first  show that by definitions of  $\ins$ and  $\del$  for every fully defined relational word $V$ we have if $V\in FDL(S)$, then all the symbols in $V$ are equal.
Then we prove by induction on the length of the word  that for every
 $n \in \mathbb{N}$ we have  $\varepsilon  \insdel^* V$  where $V$ is a fully defined relational word of length $n$ such that for every $x,y\in X^V$ we have $(x,y)\in E^V$.
Thus for every relational word $V$ we have $V\in FDL(S)$ iff for every $x,y \in X^V$ we have $(x,y)\in E^V$. The full proof of the lemma can be found in the appendix.


In the next lemma we show that if $S\in I_2D_3 \cup I_3D_2$ and either $I$ or $D$ contains unequal symbols then there is a constant $k\in \mathbb{N}$ such that for every relational word $V$ that could be obtained in $S$ from the empty word, the length of the longest fully defined scattered subword of $V$ is not larger then $k$.

\begin{lemma} \label{lem:5} 
Let $S=(\{I\}, \{D\})$ be a simple insertion-deletion system such that $S\in I_2D_3 \cup I_3D_2$ and there are $x,y \in X^I$ such that  $(x,y)\in N^I$, or there are $x',y' \in X^D$ such that  $(x',y')\in N^D$. Then there is $k\in \mathbb{N}$ such that for every relational word $V$ if $V\in L(S)$, then $maxFD(V)\leq k$.
\end{lemma}
\begin{proof}
Let $S=(\{I\}, \{D\})$ be a simple insertion-deletion system such that $S\in I_2D_3 \cup I_3D_2$ and $V$ be a relational word such that $V\in L(S)$, i.e., $\varepsilon  \insdel^* V$. Then by Lemma~\ref{lem:1} we have that there are $V_0, V_1, V_2,..\in \mathbb{RW}$ such that $\varepsilon  \ins^* V_0 \del V_1 \del V_2 \del ...\del V$.

By the definition of the insertion relation we have that for every relational words $X$ and $Y$ if $X\ins^* Y$, then
$maxE(Y)=max (maxE(X), maxE(I))$,
$maxFD(Y)=max( maxFD(X), maxFD(I))$.
Then  it follows that
$maxE(V_0)= maxE(I)$,
$maxFD(V_0)=maxFD(I)=|I|$.

Now we consider the relational word $D$. Since either $I$ or $D$ contains
unequal symbols, there are three cases for $D$: (1) there are no equal symbols
in $D$; (2) all symbols in $D$ are equal; (3) $D$ contains both equal and
unequal symbols.


It can be shown by induction that
\begin{itemize}
\item in Case 1 there is a constant $k= max (|I|, |D|\cdot (maxE(I)-1))$ such that for every $V_i$, $i\geq 1$ we have  $maxFD(V_i)\leq max (|I|, |D|\cdot (maxE(I)-1))$. Since $S\in I_2D_3 \cup I_3D_2$, it is obvious that $maxFD(V)\leq 4$.
\item  in Case 2 for every $V_i$ we have $maxFD(V_i)\leq |I|$. Since $S\in I_2D_3 \cup I_3D_2$, we have $|I|\leq 3$, then $maxFD(V)\leq 3$.
\item  in Case 3 for every $V_i$ we have $maxFD(V_i)\leq 3$, then  $maxFD(V)\leq 3$.
\end{itemize}

The full proof of the lemma can be found in the appendix.
\end{proof}


\begin{corollary} \label{cor:6} 
If $S=(\{I\}, \{D\})$ be a simple insertion-deletion system such that $S\in I_2D_3 \cup I_3D_2$ and there are $x,y \in X^I$ such that  $(x,y)\in N^I$, or there are $x',y' \in X^D$ such that  $(x',y')\in N^D$, then the set $FDL(S)$ is finite.
\end{corollary}
\begin{proof}
It follows from Lemma~\ref{lem:5} that for every $S\in I_2D_3 \cup I_3D_2$ the number of different fully defined scattered subwords in the set of all words that could be obtained from the empty word in $S$ is finite since the length of such subwords is bounded by $k$. Then by Corollary~\ref{cor:3} the set $FDL(S)$ of all fully defined relational words that could be obtained in $S$ is finite.
\end{proof}

\begin{theorem}\label{th:7} 
Given a simple  insertion-deletion system $S\in I_2D_3 \cup I_3D_2$ and fully defined relational word $V$, it is decidable, whether  $V\in L(S)$, i.e., whether  $\varepsilon  \insdel^* V$.
\end{theorem}
\begin{proof}
Let $V$ be a fully defined relational word and $S=(\{I\}, \{D\})$ be a simple  insertion-deletion system such that  $S\in I_2D_3 \cup I_3D_2$. Then there are 2 cases:
\begin{enumerate}
  \item All symbols in both insertion and deletion rules are equal, i.e., for every $x,y \in X^I$ we have $(x,y)\in E^I$ and for every $x',y' \in X^D$ we have $(x',y')\in E^D$;
  \item Either $I$ or $D$ contains unequal symbols, i.e., there are $x,y \in X^I$ such that  $(x,y)\in N^I$, or there are $x',y' \in X^D$ such that  $(x',y')\in N^D$.
\end{enumerate}
In the first case by Lemma~\ref{lem:4} we have that $V\in L(S)$ iff for every $x,y \in X^V$ we have $(x,y)\in E^V$. Then it is obvious that it is decidable, if $V\in L(S)$.

In the second case by Lemma~\ref{lem:5} and Corollary~\ref{cor:6} we have that the set $L(S)$ is finite and there is a constant $k$ that depends only on parameters of $S$ such that each word in $FDL(S)$ is not longer than $k$. Then we can obtain all the words in $FDL(S)$ in finite time by building the derivation tree.
\end{proof}


\section{Universality}

In this section we show that if the length of the inserted and deleted words
can be large, then corresponding insertion-deletion systems can produce a
coding of any recursively enumerable language. We will abuse the terminology
and we will call a function $f:\mathcal A^*\to \RW$ (where $\mathcal A$ is an
alphabet) a morphism, if it satisfies $f(uv)=f(u)f(v)$. We will further
restrict this notion and consider only those morphisms having $f(a)\in \FDRW$,
for any $a\in\mathcal A$. Since any $w\in \FDRW$ can be uniquely identified by
a string, we will use such a representation to define corresponding morphisms.
Notice, that $f(u)\not\in \FDRW$ for $|u|>1$.

\begin{theorem}\label{thm:univ}
For any recursively enumerable language $\mathcal L$ over a finite alphabet
$\mathcal A$ and for any (possibly infinite) alphabet $\mathcal V$ with
$|\mathcal V|>2$, there exists an insertion-deletion system over relational
words $S=(\mathcal V,INS,DEL,A)$ and a morphism $h$  such that $\mathcal
L=h^{-1}(L(S))$.
\end{theorem}

\begin{proof}
It is known that any recursively enumerable language can be generated by a
context-free insertion-deletion system using strings over a finite
alphabet~\cite{Margenstern}. This can be achieved by insertion of words of
length~3 (resp. 2) and deletion of words of length~2 (resp. 3). Let
$S'=(V',T',INS',DEL',A')$ be such kind of system with $L(S')=\mathcal L$. We
recall that $L(S')$ contains words over $T'$ reachable from the axioms of $A'$.

Let $c:\mathcal A\to \FDRW$ be the morphism defined as follows:
$c(\mathfrak{a}_i)=(ab)^Ka^i(ba)^K$, $1\le i\le n$, where $n=|\mathcal A|$ and
$K>n+2$. We will call $c(\mathfrak{a}_i)$ the \emph{code} of the letter
$\mathfrak{a}_i$. We say that $w\in \RW$ is in \emph{canonical} form if
$c^{-1}(w)\ne\emptyset$. Consider the extension of $c$ to languages and let
$INS=c(INS')$, $DEL=c(DEL')$ and $A=c(A')$. We also define $h(a)=c(a)$, if
$a\in T'$.

We claim that $\mathcal L=h^{-1}(L(S))$. Clearly, due to the construction of
$S$ we immediately obtain that $L(S)$ contains the image by $c$ of all
sentential forms used to obtain a word from $L(S')$. Next, we remark that the
inverse morphism $h^{-1}$ permits to select only relational words in canonical
form corresponding to the concatenation of codes of terminal letters from $T'$,
therefore its application yields a word from $L(S')$. Thus we obtain that
$\mathcal L\subseteq h^{-1}(L(S))$.

In order to show the converse inclusion we need to prove that no other words
except those corresponding to derivations of $S'$ can be obtained. This can be
formalized as follows.

\begin{claim}
For any derivation $w_1\Rightarrow\dots{}\Rightarrow w_n$, $w_i\ne w_j$, $1\le
i< j\le n$ in $S$, if $w_n$ is in canonical form, then any $w_k$, $1\le k< n$
is in canonical form.
\end{claim}

In order to prove the above assertion we will show that having $w\Rightarrow
w'$, with $w$ being in canonical form and $w'$ not being in canonical form,
implies that $w'\not\Rightarrow^* w''$ with $w''$ being in canonical form
($w''\ne w$). We shall prove this statement by contradiction. To simplify the
arguments we  assume that the sequence of derivations leading from $w'$ to
$w''$ does not contain repeated words.
We remark that $w'$ can only be
obtained by an insertion from $w$ at a position not corresponding to the
codeword boundary. Hence, in order to obtain a canonical word it is needed to
``break'' a sequence of codewords into pieces by insertion and to reconstruct
new different codewords from these pieces. Since the deletion operation is
performed for words in canonical form, a subword in canonical form should be
obtained using the insertion operation.

We recall that each codeword $c(\mathfrak{a}_i)$ is composed from 3 different
parts: the left part -- $(ab)^K$, the middle part -- $a^i$ and the right part
-- $(ba)^K$. Since these parts can never match each other, the only way to
obtain a subword in canonical form is to construct it symbol by symbol by a
nested insertion of at most $i+2K$ codewords, each insertion should be done
after the first (or before the last) symbol of each newly inserted word. We
remark that this method permits to construct any sequence of symbols. However,
since it takes the first letter away (the case of the last letter is similar),
the remaining $i+2K$ left parts will contain the sequence $b(ab)^{K-1}$. Such a
pattern can be completed to a codeword using the method above, but this
introduces at least $2K(K-1)$ same patterns $b(ab)^{K-1}$. Another possibility
to complete it is to insert a codeword and to use its rightmost letter $a$, but
this yields to the formation of the pattern $(ba)^{K-1}b$, hence the number of
incorrect patterns does not change.

To conclude, in order to construct a pattern corresponding to a single codeword
$c(\mathfrak{a}_i)$ at least $i+2K$ ``incorrect'' patterns (that need to be
completed to a codeword) are generated. Moreover, the generation is performed
in a nested manner, so no parts of it can form a subword in canonical form.
Since each ``completion'' step introduces more words to be completed, this
process can never lead to a word in a canonical form.

Now to conclude the proof of the theorem we remark that if every derivation in
$S$ is using words in canonical form, then this directly corresponds to a
derivation in $S'$. Hence, no new words can be obtained yielding $\mathcal
L\supseteq h^{-1}(L(S))$, which concludes the proof.
\end{proof}

Since the membership problem for recursively enumerable languages is
undecidable we obtain the following corollary.

\begin{corollary}
Given an  insertion-deletion system $S=(\mathcal V,INS,DEL,A)$ and a relational
word $X$, it is undecidable, whether  $X\in L(S)$, i.e., whether $Z \insdel^*
X$, $Z\in A\cup\{\varepsilon\}$.
\end{corollary}

\noindent
{\bf \large Further remarks:}
We propose below two extensions of the model of insertion-deletion on
relational words introduced in this paper. First we remark that a rewriting
rule $u\to v$ can be seen as the deletion of $u$ and an insertion of $v$ at the
corresponding place. So, with small technical changes, the
definitions~\ref{def:insertion} and~\ref{def:deletion} can be combined into a
single definition for the rewriting operation. We remark that in the case of
rewriting, the counterpart of Theorem~\ref{thm:univ} becomes trivial as the
synchronization of the insertion and the deletion operation allows only
rewriting of adjacent codewords.

Another extension is to consider the counterpart of the contextual variants of
the insertion and deletion operation on strings~\cite{PRSbook}. In this case,
the insertion or the deletion is performed is a specific context. The
definition~\ref{def:insertion} can be adapted by first combining the left and
right contexts into a single word, using a pattern-matching step like in
definition~\ref{def:deletion} and then inserting the new word at the position
given by contexts and keeping the relations between the context and the
inserted word. For example, a rule $(a,ab,b)$ would find an occurrence of two
unequal symbols in the word and then would insert exactly between them two
symbols equal to the symbol at left (resp. right) of the current position. The
deletion operation can be defined similarly. In the case of contextual
insertion and deletion the counterpart of Theorem~\ref{thm:univ} is also
trivial, because it is possible to use the codes of entire symbols as left and
right context. This means, that the operations can only be performed if the
codewords are adjacent, i.e. in canonical form.







\appendix
\section{Appendix}

\begin{proof}[Proof of Lemma~\ref{lem:2}]
We prove  that in each $S\in I_3D_2$ we have $V  \insdel^* \varepsilon$ where $|V|=1$.  \\
Case 1. Let $S=(\{I\}, \{D\})$ where
$I=\left(
  \begin{array}{ccc}
    1 & 1 & 1 \\
    1 & 1 & 1 \\
    1 & 1 & 1 \\
  \end{array}
\right)$, $D=\left(
   \begin{array}{cc}
     1 & 1 \\
     1 & 1\\
   \end{array}
 \right)$.
In this case we have following derivation:
$V \insa[1]{I} V_1   \dela[3]{D}  V_2 \dela[1]{D} \varepsilon$, and in matrix notation we have
 $\left(
                 \begin{array}{c}
                   1 \\
                 \end{array}
               \right)
  \insa[1]{I} \left(
         \begin{array}{c|ccc|}
           1 & 2 & 2 & 2 \\
           \hline
            2 & 1 & 1 & 1 \\
           2 & 1 & \y 1 & \y 1 \\
           2 & 1 & \y 1 & \y 1 \\
           \hline
         \end{array}
       \right)
       \dela[3]{D}
       \left(
         \begin{array}{cc}
           \y 1 & \y 2 \\
           \y 2 & \y 1 \\
         \end{array}
       \right)
       \dela[1]{D}  \varepsilon
  $\\
Case 2. Let $S=(\{I\}, \{D\})$ where
 $I=\left(
  \begin{array}{ccc}
    1 & 1 & 1 \\
    1 & 1 & 1 \\
    1 & 1 & 1 \\
  \end{array}
\right)$, $D=\left(
   \begin{array}{cc}
     1 & 0 \\
     0 & 1\\
   \end{array}
 \right)$.

In this case we have $V \insa[0]{I} V_1  \insa[1]{I} V_2 \insa[1]{I} V_3 \dela[7]{D}  V_4 \dela[4]{D} V_5 \dela[3]{D}  V_6 \dela[1]{D} \varepsilon$
and in matrix notation we have $\left(
                 \begin{array}{c}
                   1 \\
                 \end{array}
               \right)
\insa[0]{I}\left(
         \begin{array}{|ccc|c}
         \hline
           1 & 1 & 1 & 2 \\
           1 & 1 & 1 & 2 \\
           1 & 1 & 1 & 2 \\
           \hline
           2 & 2 & 2 & 1 \\
         \end{array}
       \right)
\insa[1]{I} \left(
       \begin{array}{c|ccc|ccc}
         1 & 2 & 2 & 2 & 1 & 1 & 2 \\
         \hline
         2 & 1 & 1 & 1 & 2 & 2 & 2 \\
         2 & 1 & 1 & 1 & 2 & 2 & 2 \\
         2 & 1 & 1 & 1 & 2 & 2 & 2 \\
         \hline
         1 & 2 & 2 & 2 & 1 & 1 & 2 \\
         1 & 2 & 2 & 2 & 1 & 1 & 2 \\
         2 & 2 & 2 & 2 & 2 & 2 & 1 \\
       \end{array}
     \right)
\insa[1]{I}
 \left(
   \begin{array}{c|ccc|cccccc}
     1 & 2 & 2 & 2 & 2 & 2 & 2 & 1 & 1 & 2 \\
     \hline
     2 & 1 & 1 & 1 & 2 & 2 & 2 & 2 & 2 & 2 \\
     2 & 1 & 1 & 1 & 2 & 2 & 2 & 2 & 2 & 2 \\
     2 & 1 & 1 & 1 & 2 & 2 & 2 & 2 & 2 & 2 \\
     \hline
     2 & 2 & 2 & 2 & 1 & 1 & 1 & 2 & 2 & 2 \\
     2 & 2 & 2 & 2 & 1 & 1 & 1 & 2 & 2 & 2 \\
     2 & 2 & 2 & 2 & 1 & 1 & 1 & 2 & 2 & 2 \\
     1 & 2 & 2 & 2 & 2 & 2 & 2 & 1 & 1 & 2 \\
     1 & 2 & 2 & 2 & 2 & 2 & 2 & 1 & \y1 & \y2 \\
     2 & 2 & 2 & 2 & 2 & 2 & 2 & 2 & \y2 & \y1 \\
   \end{array}
 \right)
 \dela[7]{D}
  \left(
   \begin{array}{cccccccc}
     1 & 2 & 2 & 2 & 2 & 2 & 2 & 1  \\
     2 & 1 & 1 & 1 & 2 & 2 & 2 & 2  \\
     2 & 1 & 1 & 1 & 2 & 2 & 2 & 2  \\
     2 & 1 & 1 & 1 & 2 & 2 & 2 & 2 \\
     2 & 2 & 2 & 2 & 1 & 1 & 1 & 2  \\
     2 & 2 & 2 & 2 & 1 & 1 & 1 & 2  \\
     2 & 2 & 2 & 2 & 1 & 1 & \y1 & \y2  \\
     1 & 2 & 2 & 2 & 2 & 2 & \y2 & \y1 \\
   \end{array}
 \right)
 \dela[4]{D}
  \left(
   \begin{array}{cccccc}
     1 & 2 & 2 & 2 & 2 & 1  \\
     2 & 1 & 1 & 1 & 2 & 2  \\
     2 & 1 & 1 & 1 & 2 & 2  \\
     2 & 1 & 1 & \y1 & \y2 & 2 \\
     2 & 2 & 2 & \y2 & \y1 & 1  \\
     1 & 2 & 2 & 2 & 1 & 1 \\
   \end{array}
 \right)
 \dela[3]{D}
  \left(
   \begin{array}{cccc}
     1 & 2 & 2 & 2   \\
     2 & 1 & 1 & 2   \\
     2 & 1 & \y1 & \y2   \\
     2 & 2 & \y2 & \y1  \\
   \end{array}
 \right)
 \dela[1]{D}
  \left(
   \begin{array}{cc}
     \y1 & \y2   \\
     \y2 & \y1    \\
   \end{array}
 \right)
     \del \varepsilon
     $

Case 3. Let $S=(\{I\}, \{D\})$ where  $I=\left(
  \begin{array}{ccc}
    1 & 0 & 0 \\
    0 & 1 & 0 \\
    0 & 0 & 1 \\
  \end{array}
\right)$, $D=\left(
   \begin{array}{cc}
     1 & 1 \\
     1 & 1\\
   \end{array}
 \right)$.

In this case we have $V \insa[0]{I} V_1  \insa[2]{I} V_2 \insa[2]{I} V_3 \dela[2]{D}  V_4 \dela[5]{D} V_5 \dela[3]{D}  V_6 \dela[1]{D} \varepsilon$.

 Then  in matrix notation we have $\left(
                 \begin{array}{c}
                   1 \\
                 \end{array}
               \right)
\insa[0]{I}
   \left(
     \begin{array}{|ccc|c}
     \hline
       1 & 0 & 0 & 2 \\
       0 & 1 & 0 & 2 \\
       0 & 0 & 1 & 2 \\
       \hline
       2 & 2 & 2 & 1 \\
     \end{array}
   \right)
\insa[2]{I}
     \left(
       \begin{array}{cc|ccc|cc}
         1 & 0 & 2 & 2 & 2 & 0 & 2 \\
         0 & 1 & 2 & 2 & 2 & 0 & 2 \\
         \hline
         2 & 2 & 1 & 0 & 0 & 2 & 2 \\
         2 & 2 & 0 & 1 & 0 & 2 & 2 \\
         2 & 2 & 0 & 0 & 1 & 2 & 2 \\
         \hline
         0 & 0 & 2 & 2 & 2 & 1 & 2 \\
         2 & 2 & 2 & 2 & 2 & 2 & 1 \\
       \end{array}
     \right)
\insa[2]{I}
     \left(
       \begin{array}{cc|ccc|ccccc}
         1 & 0 & 2 & 2 & 2 & 2 & 2 & 2 & 0 & 2 \\
         0 & 1 & 2 & 2 & 2 & 2 & 2 & 2 & 0 & 2 \\
         \hline
         2 & 2 & 1 & 0 & 0 & 2 & 2 & 2 & 2 & 2 \\
         2 & 2 & 0 & 1 & 0 & 2 & 2 & 2 & 2 & 2 \\
         2 & 2 & 0 & 0 & 1 & 2 & 2 & 2 & 2 & 2 \\
         \hline
         2 & 2 & 2 & 2 & 2 & 1 & 0 & 0 & 2 & 2 \\
         2 & 2 & 2 & 2 & 2 & 0 & 1 & 0 & 2 & 2 \\
         2 & 2 & 2 & 2 & 2 & 0 & 0 & \y1 & \y2 & 2 \\
         0 & 0 & 2 & 2 & 2 & 2 & 2 & \y2 & \y1 & 2 \\
         2 & 2 & 2 & 2 & 2 & 2 & 2 & 2 & 2 & 1 \\
       \end{array}
     \right)
\dela[2]{D}
          \left(
       \begin{array}{cccccccc}
         1 & 0 & 2 & 2 & 2 & 2 & 2 & 2 \\
         0 & \y1 & \y2 & 2 & 2 & 2 & 2 & 2 \\
         2 & \y2 & \y1 & 0 & 0 & 2 & 2 & 2 \\
         2 & 2 & 0 & 1 & 0 & 2 & 2 & 2 \\
         2 & 2 & 0 & 0 & 1 & 2 & 2 & 2 \\
         2 & 2 & 2 & 2 & 2 & 1 & 0 & 2 \\
         2 & 2 & 2 & 2 & 2 & 0 & 1 & 2 \\
         2 & 2 & 2 & 2 & 2 & 2 & 2 & 1 \\
       \end{array}
     \right)
\dela[5]{D}
          \left(
       \begin{array}{cccccc}
         1 & 2 & 2 & 2 & 2 & 2 \\
         2 & 1 & 0 & 2 & 2 & 2 \\
         2 & 0 & 1 & 2 & 2 & 2 \\
         2 & 2 & 2 & 1 & 0 & 2 \\
         2 & 2 & 2 & 0 & \y1 & \y2 \\
         2 & 2 & 2 & 2 & \y2 & \y1 \\
       \end{array}
     \right)
\dela[3]{D}
          \left(
       \begin{array}{cccc}
          1 & 2 & 2 & 2 \\
          2 & 1 & 0 & 2 \\
          2 & 0 & \y1 & \y2 \\
          2 & 2 & \y2 & \y1 \\
       \end{array}
     \right)
\dela[1]{D}
          \left(
       \begin{array}{cc}
           \y1 & \y2 \\
           \y2 & \y1 \\
       \end{array}
     \right)
     \del \varepsilon
  $

 Case 4. Let $I=\left(
  \begin{array}{ccc}
    1 & 0 & 0 \\
    0 & 1 & 0 \\
    0 & 0 & 1 \\
  \end{array}
\right)$, $D=\left(
   \begin{array}{cc}
     1 & 0 \\
     0 & 1\\
   \end{array}
 \right)$.

 Then we have $V \insa[0]{I} V_1   \dela[1]{D}  V_2 \dela[1]{D} \varepsilon$.

 Then in matrix notation $\left(
                 \begin{array}{c}
                   1 \\
                 \end{array}
               \right)
\insa[0]{I}
  \left(
        \begin{array}{|ccc|c}
        \hline
          1 & 0 & 0 & 2 \\
          0 & 1 & 0 & 2 \\
          0 & 0 & \y1 & \y2 \\
          \hline
          2 & 2 & \y2 & \y1 \\
        \end{array}
      \right)
\dela[1]{D}
  \left(
        \begin{array}{cc}
          \y1 & \y0  \\
          \y0 & \y1 \\
        \end{array}
      \right)
\dela[1]{D}  \varepsilon
  $

Case 5.  Let $I=\left(
  \begin{array}{ccc}
    1 & 1 & 0 \\
    1 & 1 & 0 \\
    0 & 0 & 1 \\
  \end{array}
\right)$, $D=\left(
   \begin{array}{cc}
     1 & 1 \\
     1 & 1\\
   \end{array}
 \right)$.

 Then we have $V \insa[0]{I} V_1   \dela[1]{D}  V_2 \dela[1]{D} \varepsilon$

and in matrix notation
we have $\left(
                 \begin{array}{c}
                   1 \\
                 \end{array}
               \right)
\insa[0]{I}
  \left(
        \begin{array}{|ccc|c}
        \hline
          \y1 & \y1 & 0 & 2 \\
          \y1 & \y1 & 0 & 2 \\
          0 & 0 & 1 & 2 \\
          \hline
          2 & 2 & 2 & 1 \\
        \end{array}
      \right)
 \dela[1]{D}
  \left(
        \begin{array}{cc}
          \y1 & \y1  \\
          \y1 & \y1 \\
        \end{array}
      \right)
 \dela[1]{D} \varepsilon
  $

Case 6. Let $I=\left(
  \begin{array}{ccc}
    1 & 1 & 0 \\
    1 & 1 & 0 \\
    0 & 0 & 1 \\
  \end{array}
\right)$, $D=\left(
   \begin{array}{cc}
     1 & 0 \\
     0 & 1\\
   \end{array}
 \right)$.
Then we have $V \insa[0]{I} V_1   \dela[2]{D}  V_2 \dela[1]{D} \varepsilon$ and in matrix notation
 we have $\left(
                 \begin{array}{c}
                   1 \\
                 \end{array}
               \right)
\insa[0]{I}
  \left(
        \begin{array}{|ccc|c}
        \hline
          1 & 1 & 0 & 2 \\
          1 & \y1 & \y0 & 2 \\
          0 & \y0 & \y1 & 2 \\
          \hline
          2 & 2 & 2 & 1 \\
        \end{array}
      \right)
\dela[2]{D}
  \left(
        \begin{array}{cc}
          \y1 & \y2  \\
          \y2 & \y1 \\
        \end{array}
      \right)
\dela[1]{D} \varepsilon
  $

 Case 7.  Let $I=\left(
  \begin{array}{ccc}
    1 & 0 & 0 \\
    0 & 1 & 1 \\
    0 & 1 & 1 \\
  \end{array}
\right)$, $D=\left(
   \begin{array}{cc}
     1 & 1 \\
     1 & 1\\
   \end{array}
 \right)$.

Then we have $V \insa[0]{I} V_1   \dela[2]{D}  V_2 \dela[1]{D} \varepsilon$ and in matrix notation
 we have $\left(
                 \begin{array}{c}
                   1 \\
                 \end{array}
               \right)
\insa[0]{I}
  \left(
        \begin{array}{|ccc|c}
        \hline
          1 & 0 & 0 & 2 \\
          0 & \y1 & \y1 & 2 \\
          0 & \y1 & \y1 & 2 \\
          \hline
          2 & 2 & 2 & 1 \\
        \end{array}
      \right)
 \dela[2]{D}
  \left(
        \begin{array}{cc}
          \y1 & \y2  \\
          \y2 & \y1 \\
        \end{array}
      \right)
 \dela[1]{D}\varepsilon
  $

  Case 8.  Let $I=\left(
  \begin{array}{ccc}
    1 & 0 & 0 \\
    0 & 1 & 1 \\
    0 & 1 & 1 \\
  \end{array}
\right)$, $D=\left(
   \begin{array}{cc}
     1 & 0 \\
     0 & 1\\
   \end{array}
 \right)$.

Then we have $V \insa[0]{I} V_1   \dela[3]{D}  V_2 \dela[1]{D} \varepsilon$ and in matrix notation
we have $\left(
                 \begin{array}{c}
                   1 \\
                 \end{array}
               \right)
 \insa[0]{I}
  \left(
        \begin{array}{|ccc|c}
        \hline
          1 & 0 & 0 & 2 \\
          0 & 1 & 1 & 2 \\
          0 & 1 & \y1 & \y2 \\
          \hline
          2 & 2 & \y2 & \y1 \\
        \end{array}
      \right)
 \dela[3]{D}
  \left(
        \begin{array}{cc}
          \y1 & \y0  \\
          \y0 & \y1 \\
        \end{array}
      \right)
 \dela[1]{D}  \varepsilon
  $

Case 9.  Let $I=\left(
  \begin{array}{ccc}
    1 & 0 & 1 \\
    0 & 1 & 0 \\
    1 & 0 & 1 \\
  \end{array}
\right)$, $D=\left(
   \begin{array}{cc}
     1 & 1 \\
     1 & 1\\
   \end{array}
 \right)$.

Then we have $V \insa[0]{I} V_1  \insa[0]{I} V_2 \insa[2]{I} V_3 \dela[6]{D}  V_4 \dela[5]{D} V_5 \dela[2]{D}  V_6 \dela[1]{D} \varepsilon$.
In matrix notation  we have $\left(
                 \begin{array}{c}
                   1 \\
                 \end{array}
               \right)
\insa[0]{I}
  \left(
        \begin{array}{|ccc|c}
        \hline
          1 & 0 & 1 & 2 \\
          0 & 1 & 0 & 2 \\
          1 & 0 & 1 & 2 \\
          \hline
          2 & 2 & 2 & 1 \\
        \end{array}
      \right)
\insa[0]{I}
  \left(
    \begin{array}{|ccc|cccc}
    \hline
      1 & 0 & 1 & 2 & 2 & 2 & 2 \\
      0 & 1 & 0 & 2 & 2 & 2 & 2 \\
      1 & 0 & 1 & 2 & 2 & 2 & 2 \\
      \hline
      2 & 2 & 2 & 1 & 0 & 1 & 2 \\
      2 & 2 & 2 & 0 & 1 & 0 & 2 \\
      2 & 2 & 2 & 1 & 0 & 1 & 2 \\
      2 & 2 & 2 & 2 & 2 & 2 & 1 \\
    \end{array}
  \right)
\insa[2]{I}
  \left(
    \begin{array}{cc|ccc|ccccc}
      1 & 0 & 2 & 2 & 2 & 1 & 2 & 2 & 2 & 2 \\
      0 & 1 & 2 & 2 & 2 & 0 & 2 & 2 & 2 & 2 \\
      \hline
      2 & 2 & 1 & 0 & 1 & 2 & 2 & 2 & 2 & 2 \\
      2 & 2 & 0 & 1 & 0 & 2 & 2 & 2 & 2 & 2 \\
      2 & 2 & 1 & 0 & 1 & 2 & 2 & 2 & 2 & 2 \\
      \hline
      1 & 0 & 2 & 2 & 2 & 1 & 2 & 2 & 2 & 2 \\
      2 & 2 & 2 & 2 & 2 & 2 & 1 & 0 & 1 & 2 \\
      2 & 2 & 2 & 2 & 2 & 2 & 0 & 1 & 0 & 2 \\
      2 & 2 & 2 & 2 & 2 & 2 & 1 & 0 & \y1 & \y2 \\
      2 & 2 & 2 & 2 & 2 & 2 & 2 & 2 & \y2 & \y1 \\
    \end{array}
  \right)
 \dela[6]{D}
 \left(
    \begin{array}{cccccccc}
      1 & 0 & 2 & 2 & 2 & 1 & 2 & 2  \\
      0 & 1 & 2 & 2 & 2 & 0 & 2 & 2  \\
      2 & 2 & 1 & 0 & 1 & 2 & 2 & 2 \\
      2 & 2 & 0 & 1 & 0 & 2 & 2 & 2 \\
      2 & 2 & 1 & 0 & 1 & 2 & 2 & 2 \\
      1 & 0 & 2 & 2 & 2 & \y1 & \y2 & 2  \\
      2 & 2 & 2 & 2 & 2 & \y2 & \y1 & 0  \\
      2 & 2 & 2 & 2 & 2 & 2 & 0 & 1  \\
     \end{array}
  \right)
 \dela[5]{D}
 \left(
    \begin{array}{cccccc}
      1 & 0 & 2 & 2 & 2 & 2  \\
      0 & 1 & 2 & 2 & 2 & 2  \\
      2 & 2 & 1 & 0 & 1 & 2 \\
      2 & 2 & 0 & 1 & 0 & 2 \\
      2 & 2 & 1 & 0 & \y1 & \y2 \\
      2 & 2 & 2 & 2 & \y2 & \y1  \\
     \end{array}
  \right)
 \dela[2]{D}
 \left(
    \begin{array}{cccc}
      1 & 0 & 2 & 2  \\
      0 & \y1 & \y2 & 2 \\
      2 & \y2 & \y1 & 0 \\
      2 & 2 & 0 & 1  \\

     \end{array}
  \right)
 \dela[1]{D}
 \left(
    \begin{array}{cc}
      \y1 & \y2  \\
      \y2 & \y1  \\
     \end{array}
  \right)
   \del \varepsilon
  $

Case 10.   Let $I=\left(
  \begin{array}{ccc}
    1 & 0 & 1 \\
    0 & 1 & 0 \\
    1 & 0 & 1 \\
  \end{array}
\right)$, $D=\left(
   \begin{array}{cc}
     1 & 0 \\
     0 & 1\\
   \end{array}
 \right)$.

Then we have $V \insa[0]{I} V_1   \dela[1]{D}  V_2 \dela[1]{D} \varepsilon$.

 Then in matrix notation we have $\left(
                 \begin{array}{c}
                   1 \\
                 \end{array}
               \right)
\insa[0]{I}
  \left(
        \begin{array}{|ccc|c}
        \hline
          \y1 & \y0 & 1 & 2 \\
          \y0 & \y1 & 0 & 2 \\
          1 & 0 & 1 & 2 \\
          \hline
          2 & 2 & 2 & 1 \\
        \end{array}
      \right)
\dela[1]{D}
  \left(
        \begin{array}{cc}
          \y1 & \y2  \\
          \y2 & \y1 \\
        \end{array}
      \right)
\dela[1]{D} \varepsilon
  $

\end{proof}

\begin{proof}[Proof of Lemma~\ref{lem:4}]

First we show that for every fully defined relational word $V$ we have if $V\in FDL(S)$, then all the symbols in $V$ are equal.

By Lemma~\ref{lem:1} for every $V$ we have that $\varepsilon  \insdel^* V$ iff there is $V'\in\mathbb{RW}$ such that $\varepsilon \ins ^* V' \del ^* V$.
Since in $I$ all symbols are equal, by the definition of $\ins$ we have that if a relational word $V'$ is obtained from $\varepsilon$ by any number of insertions of $I$, then there are no inequal symbols in $V'$.
Since all the symbols in $D$ are also equal, by the definition of $\del$ while applying the deletion
rule we cannot define that some symbols are inequal.  Thus if $\varepsilon  \insdel^* V$ and $V$ is fully defined , then for every $x,y\in X^V$ we have $(x,y)\in E^V$.

 Now we prove by induction on the length of the word  that for  every $n \in \mathbb{N}$ we have  $\varepsilon  \insdel^* V$  where $V$ is a fully defined relational word of length $n$ such that for every $x,y\in X^V$ we have $(x,y)\in E^V$.

 For the induction base, we suppose that $n = 1$. If $S\in I_2D_3$, i.e., the rule $I$ consists of two equal symbols and the rule $D$ consists of three equal symbols, then to derive $V$ from $\varepsilon$  we have to apply the insertion rule twice and then the deletion rule once. The illustration for this case is Fig. 6.

\begin{figure}[h]
\centering
\includegraphics[scale=0.5]{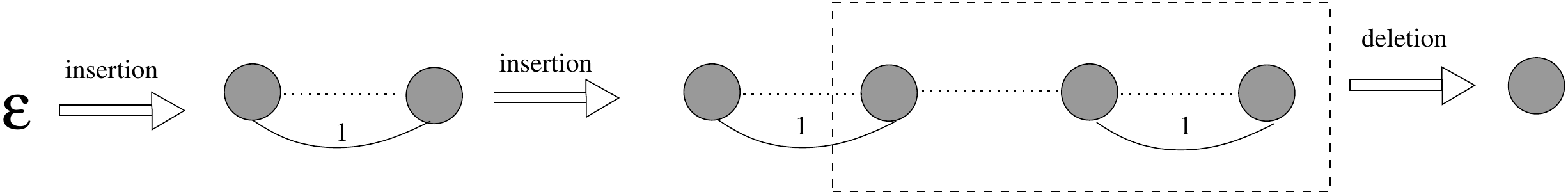}
\caption{ Obtaining $V$ from  $\varepsilon$ when $|V|=1$ and $S\in I_2D_3$}
\end{figure}

 If $S\in I_3D_2$, i.e., the rule $I$ consists of three equal symbols and the rule $D$ consists of two equal symbols, then to derive $V$ from $\varepsilon$  we have to apply the insertion rule once and then the deletion rule once. The illustration for this case is Fig. 7.

 \begin{figure}[h]
\centering
\includegraphics[scale=0.5]{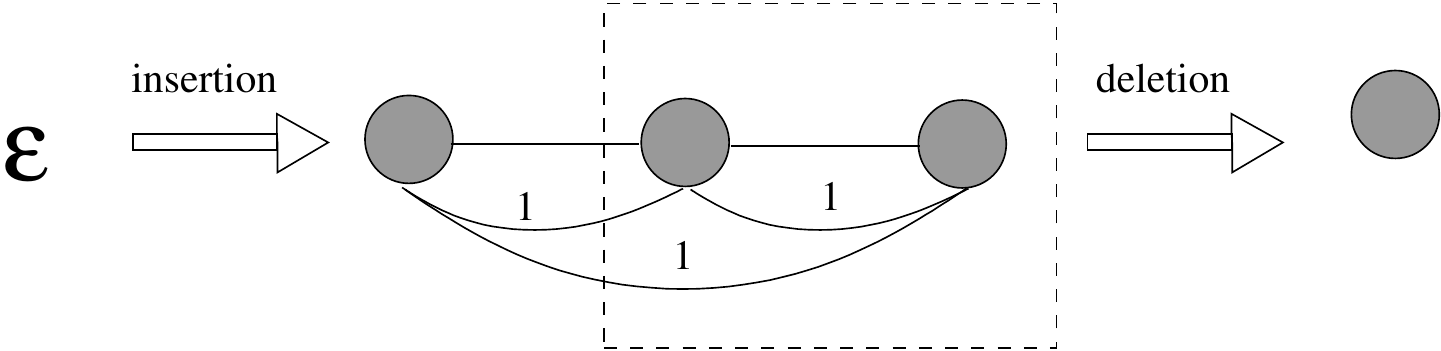}
\caption{Obtaining $V$ from  $\varepsilon$ when $|V|=1$ and $S\in I_3D_2$}
\end{figure}

 For the induction step, we assume that the hypothesis is true for all words of length $n$ or less.

 Let $V$ be a word of length $n$. If $S\in I_3D_2$, then we make one insertion at the end of the word $V$ and obtain the word of length $n + 3$. After that we delete the subword of two symbols - the last symbol of the word $V$ and the first symbol of those which we have just inserted. The relation between these symbols is undefined, so we define them to be equal. Then by the definition of the deletion rule we define that all new symbols are equal to all the symbols of word $V$. After deleting two symbols we obtain a new word of length $n+1$ where all symbols are equal. This case is shown on Fig. 8.

\begin{figure}[h]
\centering
\includegraphics[scale=0.5]{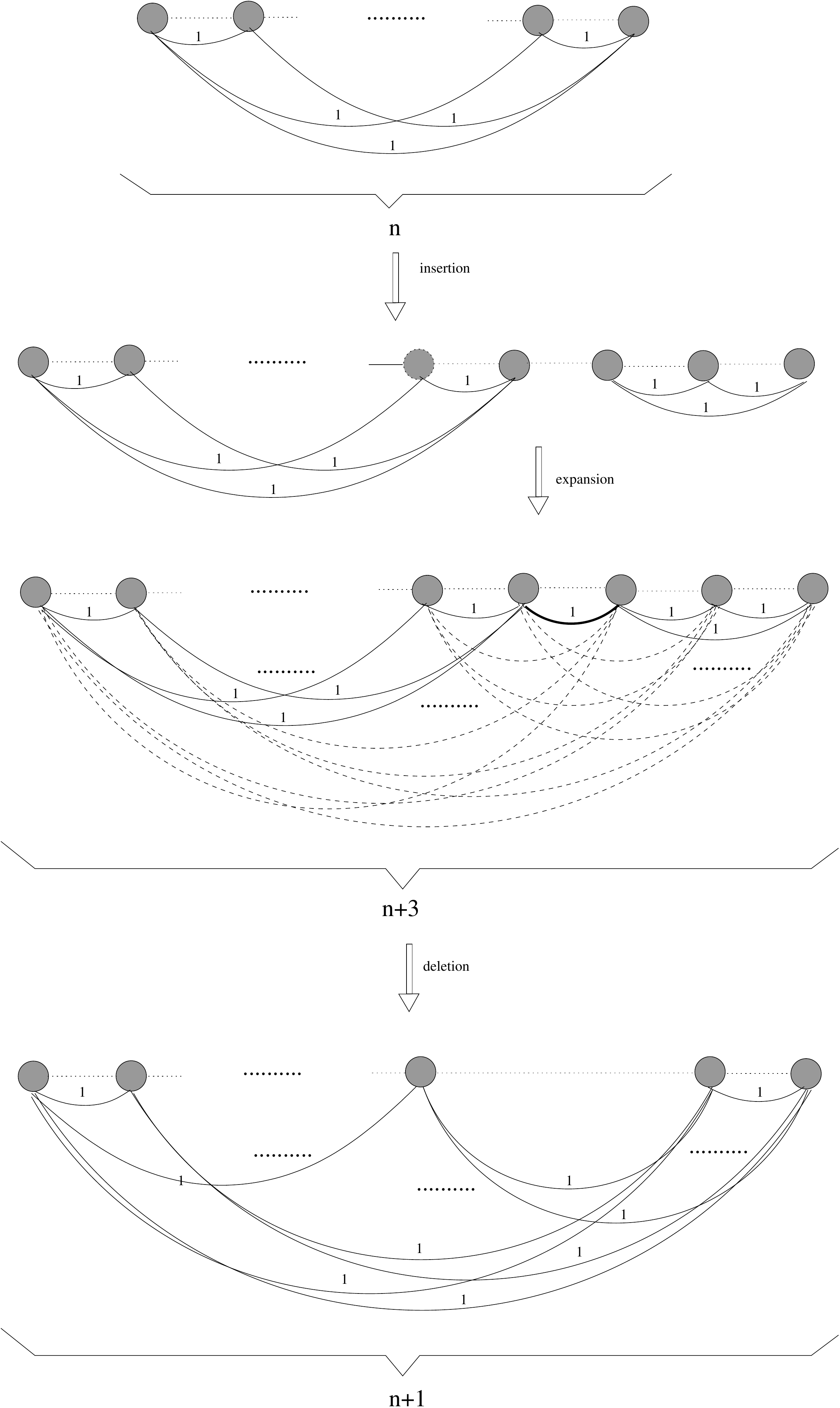}
\caption{}
\end{figure}

If $S\in I_2D_3$, then we insert two equal symbols at the end of the word $V$, and then insert another two between them. After that we delete the subword of length three that starts from the last symbol of the word $V$, defining all undefined symbols to be equal. Then after deletion we obtain a new word of length $n + 1$ where all symbols are equal. This case is shown on Fig. 9. This completes the induction step.

\begin{figure}[h]
\centering
\includegraphics[scale=0.5]{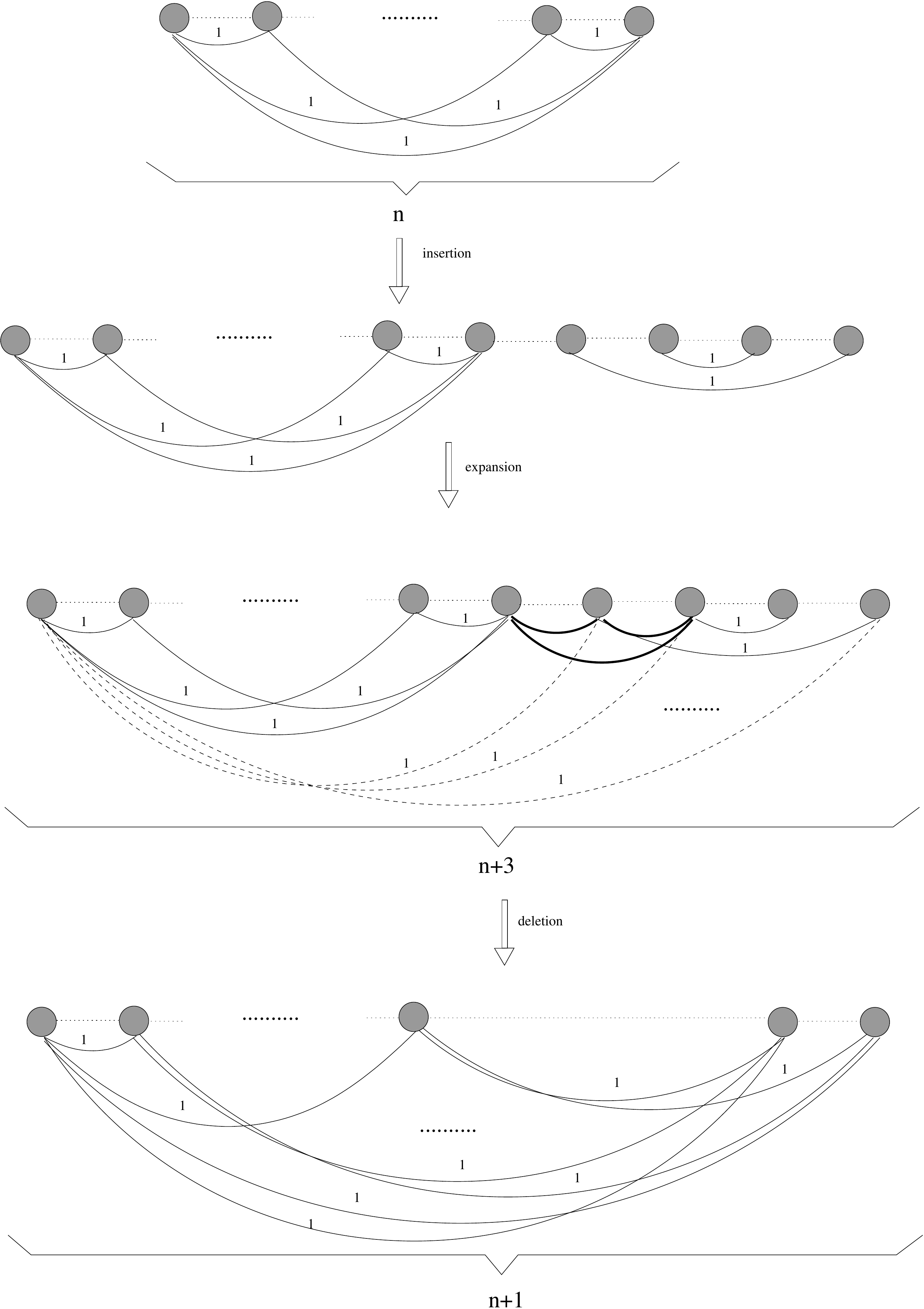}
\caption{}
\end{figure}

\end{proof}

\begin{proof}[Proof of Lemma~\ref{lem:5}]
Let $S=(\{I\}, \{D\})$ be a simple insertion-deletion system such that $S\in I_2D_3 \cup I_3D_2$ and $V$ be a relational word such that $V\in L(S)$, i.e., $\varepsilon  \insdel^* V$. Then by Lemma~\ref{lem:1} we have that there are $V_0, V_1, V_2,..\in \mathbb{RW}$ such that $\varepsilon  \ins^* V_0 \del V_1 \del V_2 \del ...\del V$.

By the definition of the insertion relation we have that for every relational words $X$ and $Y$ if $X\ins^* Y$, then
\begin{equation} \label{eq:maxE }
maxE(Y)=max (maxE(X), maxE(I))
\end{equation}
\begin{equation} \label{eq: maxFD}
maxFD(Y)=max( maxFD(X), maxFD(I))
\end{equation}
Then  we have
\begin{equation} \label{eq:maxE(V_0)}
maxE(V_0)= maxE(I)
\end{equation}
\begin{equation} \label{eq: maxFD(V_0)}
maxFD(V_0)=maxFD(I)=|I|
\end{equation}

Now we consider the relational word $D$.
Since either $I$ or $D$ contains unequal symbols, there are three cases for $D$:
\begin{enumerate}
  \item There are no equal symbols in $D$
  \item All symbols in $D$ are equal
  \item $D$ contains both equal and unequal symbols
\end{enumerate}

\textbf{Case 1.} By the definition of the deletion relation, we have that if there are no equal symbols in $D$, i.e., $maxE(D)=1$, then during expansion step we never define that some symbols are equal.

Hence, for every relational words $X$ and $Y$ if $X\del Y$, then
$maxE(Y)\leq maxE(X)$.

Then for every $V_i$ we have $maxE(V_i)\leq maxE(V_0)$ and thus by \eqref{eq:maxE(V_0)}
\begin{equation}
maxE(V_i)\leq maxE(I) \label{eq:maxE(V_i)-1}
\end{equation}
Since there are no equal symbols in $D$, the longest possible fully defined subword of $Y$ that we can obtain when we apply the deletion rule $X\del Y$ is of length $|D|\cdot (maxE(X)-1)$  (in this case there are $|D|$ groups of equal symbols of size $maxE(X)$  in the relational word $X$, and these groups are located in such a way that by applying one-step deletion we can define that the symbols of these groups are not equal).
Thus if $X\del Y$, then
\begin{equation}\label{eq:maxFD-1}
maxFD(Y)\leq max (maxFD(X) , |D|\cdot (maxE(X)-1))
\end{equation}

Then it could be easily shown by induction that for every $V_i$, $i\geq 1$ we have
\begin{equation}
maxFD(V_i)\leq max (|I|, |D|\cdot (maxE(I)-1)).
\end{equation}

\textsc{Base step.} Let $i=1$, then by \eqref{eq:maxFD-1} we have
\begin{equation}maxFD(V_1)\leq max (maxFD(V_0) , |D|\cdot (maxE(V_0)-1)).\end{equation}
Since by \eqref{eq: maxFD(V_0)}  $maxFD(V_0)=|I|$ and by \eqref{eq:maxE(V_0)} $maxE(V_0)=maxE(I)$, we have \\
$maxFD(V_1)\leq max(|I|, |D|\cdot (maxE(I)-1))$.

\textsc{Induction step. } Let us assume that there is $i \geq 1$ such that for every $j\leq i$ we have $maxFD(V_j)\leq max (|I|, |D|\cdot (maxE(I)-1))$. \\
Then  by \eqref{eq:maxFD-1} we have $maxFD(V_{i+1})\leq max (maxFD(V_i) , |D|\cdot (maxE(V_i)-1))$. \\
Since  by \eqref{eq:maxE(V_i)-1}  $maxE(V_i) \leq maxE(I)$, we have\\
$maxFD(V_{i+1})\leq max( max (|I|, |D|\cdot (maxE(I)-1)), |D|\cdot (maxE(I)-1) )$, i.e.  \\
$maxFD(V_{i+1})\leq max (|I|, |D|\cdot (maxE(I)-1))$. This completes the induction step.

Thus for every $V_i$ we have $maxFD(V_i)\leq max (|I|, |D|\cdot (maxE(I)-1))$ and hence there is a constant $k= max (|I|, |D|\cdot (maxE(I)-1))$ such that   $maxFD(V)\leq k$. Since $S\in I_2D_3 \cup I_3D_2$, it is obvious that $maxFD(V)\leq 4$.

\textbf{Case 2. }In a similar way it can be shown that in the case when all the symbols in the deletion rule $D$ are equal, for every $i$  we have
\begin{equation} \label{eq:2-1}
maxE(V_{i+1})\leq max (maxE(V_i), |D|\cdot (maxE(V_i)-1)),
\end{equation}
\begin{equation} \label{eq:2-2}
maxFD(V_{i+1})\leq max (maxFD(V_i), maxFD(V_i)+ (|D|-1)\cdot maxE(V_i)-|D| ).
\end{equation}
We show by induction that for every $V_i$ we have $maxFD(V_i)\leq |I|$.

\textsc{Base step.} Let $i=1$. By \eqref{eq:2-1} we have $maxE(V_1)\leq max (maxE(V_0), |D|\cdot (maxE(V_0)-1))$. Since $S\in I_2D_3 \cup I_3D_2$, we have that either $S\in I_2D_3$ and $|D|=3$ or $S\in I_3D_2$ and $|D|=2$.

If $|D|=2$  and all the symbols in the deletion rule $D$ equal, then $|I|=3$ and $maxE(I)\leq 2$. Then $maxE(V_1)\leq max (maxE(I), 2) \leq 2$.

 If $|D|=3$  and all the symbols in the deletion rule $D$ equal, then $|I|=2$ and $maxE(I)\leq 1$. Then $maxE(V_1)\leq maxE(I)\leq 1$.

By \eqref{eq:2-2} we have  $maxFD(V_1)\leq max (maxFD(V_0), maxFD(V_0)+ (|D|-1)\cdot maxE(V_0)-|D| )$.

If $|D|=2$, then $maxE(V_0)\leq 2$ and  $maxFD(V_1)\leq maxFD(V_0)\leq |I|$.

If $|D|=3$, then $maxE(V_0)\leq 1$  and again $maxFD(V_1)\leq maxFD(V_0)\leq |I|$.

\textsc{Induction step. } Let us assume that there is $i \geq 1$ such that for every $j\leq i$ we have $maxFD(V_j)\leq |I|$.

Since by \eqref{eq:2-1} $maxE(V_{i+1})\leq max (maxE(V_i), |D|\cdot (maxE(V_i)-1))$, we have that if $|D|=2$, then $maxE(V_i)\leq 2$, and if $|D|=3$, then $maxE(V_i)\leq 1$.

 Since  by \eqref{eq:2-2}  $maxFD(V_{i+1})\leq max (maxFD(V_i), maxFD(V_i)+ (|D|-1)\cdot maxE(V_0)-|D| )$, we have that
if $|D|=2$, then $maxFD(V_{i+1})\leq |I|$, and
if $|D|=3$,  then again $maxFD(V_{i+1})\leq |I|$.

Thus there is a constant $k= |I|$ such that  $maxFD(V)\leq k$.
Since $S\in I_2D_3 \cup I_3D_2$, we have $|I|\leq 3$, then $maxFD(V)\leq 3$.

\textbf{Case 3. }
 Since $S\in I_2D_3 \cup I_3D_2$ and $D$ contains both equal and unequal symbols, then  it is possible only when $|D|=3$ and $D$ contain two equal symbols and the third symbol is not equal to them.
In this case it can be shown that for every $i$  we have
\begin{equation} \label{eq:3-1}
maxE(V_{i+1})\leq max (maxE(V_i), 2\cdot (maxE(V_i)-1)),
\end{equation}
\begin{equation} \label{eq:3-2}
maxFD(V_{i+1})\leq max (maxFD(V_i),\\ maxFD(V_i)+ maxE(V_i)-2, 3\cdot(maxE(V_i)-1) ).
\end{equation}
 Since $|D|=3$, it follows that $|I|=2$, i.e., $maxE(V_0)\leq 2$ and $maxFD(V_0)\leq 2$.

We show by induction that for every $i$ we have  $maxFD(V_i)\leq 3$.

\textsc{Base step.} Let $i=1$. Then we have $maxE(V_1)\leq max (maxE(V_0), \\2\cdot (maxE(V_0)-1))$, i.e., $maxE(V_1)\leq 2$.

Since
$maxFD(V_1)\leq max (maxFD(V_0), maxFD(V_0)+ maxE(V_0)-2, 3\cdot(maxE(V_0)-1) )$, it follows that $maxFD(V_1)\leq 3$.

\textsc{Induction step. }  Let us assume that there is $i \geq 1$ such that for every $j\leq i$ we have $maxFD(V_j)\leq 3$.

Since  for every $i$  we have $maxE(V_{i+1})\leq max (maxE(V_i), 2\cdot (maxE(V_i)-1))$ and $maxE(V_0)\leq 2$, it follows that for every $i$ we have $maxE(V_i)\leq 2$.

Since $maxFD(V_{i+1})\leq max (maxFD(V_i), maxFD(V_i)+ maxE(V_i)-2, 3\cdot(maxE(V_i)-1) )$, we have   $maxFD(V_{i+1})\leq max (3, 3+ maxE(V_i)-2, 3\cdot(maxE(V_i)-1) )$.
Then $maxFD(V_{i+1})\leq 3$.
Thus there is a constant $k= 3$ such that  $maxFD(V)\leq k$.

\end{proof}

\end{document}